\RequirePackage{mathtools}
\documentclass[online]{jfp1}

\usepackage{hyperref}
\usepackage[all]{xy}
\usepackage[usenames,dvipsnames]{color}
\usepackage{soul}
\usepackage{stmaryrd}

\makeatletter
\@ifundefined{lhs2tex.lhs2tex.sty.read}%
  {\@namedef{lhs2tex.lhs2tex.sty.read}{}%
   \newcommand\SkipToFmtEnd{}%
   \newcommand\EndFmtInput{}%
   \long\def\SkipToFmtEnd#1\EndFmtInput{}%
  }\SkipToFmtEnd

\newcommand\ReadOnlyOnce[1]{\@ifundefined{#1}{\@namedef{#1}{}}\SkipToFmtEnd}
\usepackage{amstext}
\usepackage{amssymb}
\usepackage{stmaryrd}
\DeclareFontFamily{OT1}{cmtex}{}
\DeclareFontShape{OT1}{cmtex}{m}{n}
  {<5><6><7><8>cmtex8
   <9>cmtex9
   <10><10.95><12><14.4><17.28><20.74><24.88>cmtex10}{}
\DeclareFontShape{OT1}{cmtex}{m}{it}
  {<-> ssub * cmtt/m/it}{}

\DeclareFontShape{OT1}{cmtt}{bx}{n}
  {<5><6><7><8>cmtt8
   <9>cmbtt9
   <10><10.95><12><14.4><17.28><20.74><24.88>cmbtt10}{}
\DeclareFontShape{OT1}{cmtex}{bx}{n}
  {<-> ssub * cmtt/bx/n}{}

\newcommand{\Conid}[1]{\mathit{#1}}
\newcommand{\Varid}[1]{\mathit{#1}}
\newcommand{\anonymous}{\kern0.06em \vbox{\hrule\@width.5em}}
\newcommand{\plus}{\mathbin{+\!\!\!+}}
\newcommand{\bind}{\mathbin{>\!\!\!>\mkern-6.7mu=}}

\usepackage{polytable}

\@ifundefined{mathindent}%
  {\newdimen\mathindent\mathindent\leftmargini}%
  {}%

\def\resethooks{%
  \global\let\SaveRestoreHook\empty
  \global\let\ColumnHook\empty}
\newcommand*{\savecolumns}[1][default]%
  {\g@addto@macro\SaveRestoreHook{\savecolumns[#1]}}
\newcommand*{\restorecolumns}[1][default]%
  {\g@addto@macro\SaveRestoreHook{\restorecolumns[#1]}}
\newcommand*{\aligncolumn}[2]%
  {\g@addto@macro\ColumnHook{\column{#1}{#2}}}

\resethooks

\newcommand{\onelinecommentchars}{\quad-{}- }
\newcommand{\commentbeginchars}{\enskip\{-}
\newcommand{\commentendchars}{-\}\enskip}

\newcommand{\visiblecomments}{%
  \let\onelinecomment=\onelinecommentchars
  \let\commentbegin=\commentbeginchars
  \let\commentend=\commentendchars}

\newcommand{\invisiblecomments}{%
  \let\onelinecomment=\empty
  \let\commentbegin=\empty
  \let\commentend=\empty}

\visiblecomments

\newlength{\blanklineskip}
\newcommand{\hsindent}[1]{\quad}
\let\hspre\empty
\let\hspost\empty

\EndFmtInput
\makeatother
\ReadOnlyOnce{polycode.fmt}%
\makeatletter

\newcommand{\hsnewpar}[1]%
  {{\parskip=0pt\parindent=0pt\par\vskip #1\noindent}}

\newcommand{\hscodestyle}{}

\newcommand{\sethscode}[1]%
  {\expandafter\let\expandafter\hscode\csname #1\endcsname
   \expandafter\let\expandafter\endhscode\csname end#1\endcsname}

  {\par\noindent
   \advance\leftskip\mathindent
   \hscodestyle
   \let\\=\@normalcr
   \let\hspre\(\let\hspost\)%
   \pboxed}%
  {\endpboxed\)%
   \par\noindent
   \ignorespacesafterend}

  {\hsnewpar\abovedisplayskip
   \advance\leftskip\mathindent
   \hscodestyle
   \let\hspre\(\let\hspost\)%
   \pboxed}%
  {\endpboxed%
   \hsnewpar\belowdisplayskip
   \ignorespacesafterend}

  {\hsnewpar\abovedisplayskip
   \advance\leftskip\mathindent
   \hscodestyle
   \let\\=\@normalcr
   \(\pboxed}%
  {\endpboxed\)%
   \hsnewpar\belowdisplayskip
   \ignorespacesafterend}

\newcommand{\plainhs}{\sethscode{plainhscode}}

\plainhs

  {\hsnewpar\abovedisplayskip
   \advance\leftskip\mathindent
   \hscodestyle
   \let\\=\@normalcr
   \(\parray}%
  {\endparray\)%
   \hsnewpar\belowdisplayskip
   \ignorespacesafterend}

  {\parray}{\endparray}

  {\(\parray}{\endparray\)}

\def\codeframewidth{\arrayrulewidth}
\RequirePackage{calc}

  {\parskip=\abovedisplayskip\par\noindent
   \hscodestyle
   \arrayrulewidth=\codeframewidth
   \tabular{@{}|p{\linewidth-2\arraycolsep-2\arrayrulewidth-2pt}|@{}}%
   \hline\framedhslinecorrect\\{-1.5ex}%
   \let\endoflinesave=\\
   \let\\=\@normalcr
   \(\pboxed}%
  {\endpboxed\)%
   \framedhslinecorrect\endoflinesave{.5ex}\hline
   \endtabular
   \parskip=\belowdisplayskip\par\noindent
   \ignorespacesafterend}

\newcommand{\framedhslinecorrect}[2]%
  {#1[#2]}

  {\(\def\column##1##2{}%
   \let\>\undefined\let\<\undefined\let\\\undefined
   \newcommand\>[1][]{}\newcommand\<[1][]{}\newcommand\\[1][]{}%
   \def\fromto##1##2##3{##3}%
   }{\) }%

  {\let\orighscode=\hscode
   \let\origendhscode=\endhscode
   \def\endhscode{\def\hscode{\endgroup\def\@currenvir{hscode}\\}\begingroup}
   \orighscode\def\hscode{\endgroup\def\@currenvir{hscode}}}%
  {\origendhscode
   \global\let\hscode=\orighscode
   \global\let\endhscode=\origendhscode}%

\makeatother
\EndFmtInput
\ReadOnlyOnce{forall.fmt}%
\makeatletter

\let\HaskellResetHook\empty
\newcommand*{\AtHaskellReset}[1]{%
  \g@addto@macro\HaskellResetHook{#1}}
\newcommand*{\HaskellReset}{\HaskellResetHook}

\newcommand\hsforall{\global\let\hsdot=\hsperiodonce}
\newcommand*\hsperiodonce[2]{#2\global\let\hsdot=\hscompose}
\newcommand*\hscompose[2]{#1}

\AtHaskellReset{\global\let\hsdot=\hscompose}

\HaskellReset

\makeatother
\EndFmtInput

\newcommand{\Nat}{\mathbb{N}}

\numberwithin{equation}{section}

\newtheorem{theorem}{Theorem}[section]
\newtheorem{defn}[theorem]{Definition}
\newtheorem{prop}[theorem]{Proposition}
\newtheorem{lemma}[theorem]{Lemma}
\newtheorem{corollary}[theorem]{Corollary}
\newtheorem{example}[theorem]{Example}
\newtheorem{remark}[theorem]{Remark}

\newcommand{\comment}[1]{\quad\{ \text{ #1 } \}}
\newenvironment{calculation}{\par\medskip\noindent\(\displaystyle}{\)\bigskip}

\newcommand{\Set}{\mathbf{\mathsf{Set}}}
\newcommand{\cat}[1]{\mathcal{#1}}
\newcommand{\cA}{\cat{A}}  
\newcommand{\cC}{\cat{C}}  
\newcommand{\cD}{\cat{D}}  
\newcommand{\cP}{\cat{P}}  
\newcommand{\op}[1]{{#1}^{\text{op}}}
\newcommand{\Mon}{\cat{M}\!\!on} 
\newcommand{\FC}[2]{{#2}^{#1}}  

\newcommand{\Hom}[3]{ {#2} \xrightarrow{#1} {#3} }
\newcommand{\tto}{\rightarrow} 

\newcommand{\Id}{I}  
\newcommand{\ladj}[1]{\lfloor {#1} \rfloor} 
\newcommand{\radj}[1]{\lceil {#1} \rceil} 
\newcommand{\App}{\cat{A}} 
\newcommand{\FN}[1]{\textsf{#1}}
\newcommand{\traverse}{\FN{traverse}}

\title[A Representation Theorem for Second-Order Functionals]
    {A Representation Theorem for Second-Order Functionals}

\author[M.\ Jaskelioff and R.\ O'Connor]
   { MAURO JASKELIOFF\\
    CIFASIS, CONICET, Argentina \\
    FCEIA, Universidad Nacional de Rosario, Argentina\\[8pt] 
    RUSSELL O'CONNOR\\
    Google Canada\\
    Kitchener, Ontario, Canada}

\begin{document}

\maketitle

\begin{abstract}
  Representation theorems relate seemingly complex
  objects to concrete, more tractable ones.  

  In this paper, we take advantage of the abstraction power of
  category theory and provide a datatype-generic representation
  theorem. More precisely, we prove a representation theorem for a
  wide class of second-order functionals which are polymorphic over a
  class of functors.
  Types polymorphic over a class of functors are easily representable
  in languages such as Haskell, but are difficult to analyse and
  reason about. The concrete representation provided by the theorem is
  easier to analyse, but it might not be as convenient to
  implement. Therefore, depending on the task at hand, the change of
  representation may prove valuable in one direction or the other.

  We showcase the usefulness of the representation theorem with a
  range of examples. Concretely, we show how the representation
  theorem can be used to prove that traversable functors are finitary
  containers, how coalgebras of a parameterised store comonad
  relate to very well-behaved
  lenses, and how algebraic effects might be implemented in a
  functional language.


\end{abstract}


\section{Introduction}

When dealing with a type which uses advanced features of modern type
systems such as polymorphism and higher-order types and functions, it
is convenient to analyse whether there is another datatype that can
represent it, as the alternative representation might be easier to
program or to reason about.
A simple example of a datatype that might be better understood through
a different representation is the type of polymorphic functions
$\forall A.\ A \to A$ which, although it involves a function space
and a universal quantifier, has only one non-bottom inhabitant: the identity function.

Hence, a representation theorem opens the design space for programmers
and computer scientists, providing and connecting different views on
some construction.
When a representation is an isomorphism, we say that it is
\emph{exact}, and the change of representation can be done in both
  directions.

In this article we will consider second-order functionals that are
polymorphic over a class of functors, such as monads or applicative
functors.
In particular we will give a concrete representation for inhabitants
of types of the form
\[
\ensuremath{\forall \Conid{F}\hsforall .~(\Conid{A}_{1}\to \Conid{F}\;\Conid{B}_{1})\to (\Conid{A}_{2}\to \Conid{F}\;\Conid{B}_{2})\to \ldots\to \Conid{F}\;\Conid{C}}
\]%
Here \ensuremath{\Conid{A}_i}, \ensuremath{\Conid{B}_i}, and $C$ are fixed types, and $F$ ranges over an
appropriate class of functors. There is a condition on the class of
functors which will be made precise during the presentation of the
theorem, but basically it amounts to the existence of free
constructions. The representation is exact, as it is an isomorphism.

We will express the representation theorem using
category theory. Although the knowledge of category
theory that is required should be covered by an introductory textbook such as~\cite{awodey2006ct},
we introduce the more important concepts in Section~\ref{sec:catprelim}.
The usefulness of the representation theorem
(Section~\ref{sec:representation}) is illustrated with a range of
examples.
Concretely, we show how coalgebras of a specific parameterised comonad
are related to very
well-behaved lenses (Section~\ref{sec:pcomonads_and_lenses}), and how
traversable functors, subjected to certain coherence laws, are exactly
the finitary containers (Section~\ref{sec:traversals}). Finally we
show how the representation theorem can help when implementing free
theories of algebraic effects~(Section~\ref{sec:algebraic_theories})
and discuss related work~(Section~\ref{sec:related}).

There is a long tradition of categorically inspired functional
programming~\cite{BdM97} even though functional programming languages
like Haskell usually lack some basic structure such as products or
coproducts. The implementation of our results in Haskell, as shown in
Section~\ref{sec:lenses_in_Haskell} and
Section~\ref{sec:algebraic_theories}, should be taken simply as
categorically-inspired code. Nevertheless, the code could be interpreted to
be ``morally correct'' in a precise technical
sense~\cite{Danielsson:2006}.

\subsection{A taste of the representation theorem}

In order to get a taste of the representation theorem, we reason
informally on a total polymorphic functional language.
Consider the type $$T=\ensuremath{\forall \Conid{F}\hsforall \colon\FN{Functor}.~(\Conid{A}\to \Conid{F}\;\Conid{B})\to \Conid{F}\;\Conid{C}}.$$%
What do the inhabitants of this type look like?

The inhabitants of $T$ are functions $h = \lambda g.\, r$.  Given that
the functor $F$ is universally quantified, the only way of obtaining a
result in $\ensuremath{\Conid{F}\;\Conid{C}}$ is that in the expression $r$ there is an
application of the argument $g$ to some $a: A$. This yields
something in \ensuremath{\Conid{F}\;\Conid{B}} rather than the sought \ensuremath{\Conid{F}\;\Conid{C}}, so a function \ensuremath{\Varid{k}\colon\Conid{B}\to \Conid{C}} is needed in order to construct a map $F(k) : F\,B \to F\, C$.
This informal argument suggests that all inhabitants of $T$ can be
built from a pair of an element of $A$ and a function \ensuremath{\Conid{B}\to \Conid{C}}.
Hence, it is natural to propose the type $A\times(\ensuremath{\Conid{B}\to \Conid{C}})$ as a
simpler representation of the inhabitants of type \ensuremath{\Conid{T}}.

More formally, in order to check that the inhabitants of $T$ are in a
one-to-one correspondence with the inhabitants of $A\times(\ensuremath{\Conid{B}\to \Conid{C}})$,
we want to find an isomorphism
$$
\ensuremath{\forall \Conid{F}\hsforall \colon\FN{Functor}.~(\Conid{A}\to \Conid{F}\;\Conid{B})\to \Conid{F}\;\Conid{C}} 
\xymatrix{~ \ar@/^/@<1ex>[r]^{\varphi}
               \ar@{}[r]|{\cong}
& 
     \ar@/^/@<1ex>[l]^{\varphi^{-1}}~} 
A\times(\ensuremath{\Conid{B}\to \Conid{C}}).
$$%
We define $\varphi^{-1}$ using the procedure described above.

\begin{hscode}\SaveRestoreHook
\column{B}{@{}>{\hspre}l<{\hspost}@{}}%
\column{7}{@{}>{\hspre}l<{\hspost}@{}}%
\column{13}{@{}>{\hspre}l<{\hspost}@{}}%
\column{24}{@{}>{\hspre}l<{\hspost}@{}}%
\column{E}{@{}>{\hspre}l<{\hspost}@{}}%
\>[B]{}\varphi^{-1}{}\<[7]%
\>[7]{}\colon\Conid{A}\times(\Conid{B}\to \Conid{C}){}\<[24]%
\>[24]{}\to \forall \Conid{F}\hsforall \colon\FN{Functor}.~(\Conid{A}\to \Conid{F}\;\Conid{B})\to \Conid{F}\;\Conid{C}{}\<[E]%
\\
\>[B]{}\varphi^{-1}\;(\Varid{a},\Varid{k}){}\<[13]%
\>[13]{}\mathrel{=}\lambda \Varid{g}.~\Conid{F}\;(\Varid{k})\;(\Varid{g}\;\Varid{a}){}\<[E]%
\ColumnHook
\end{hscode}\resethooks
In order to define \ensuremath{\varphi}, notice that \ensuremath{\Conid{R}\;\Conid{C}\mathrel{=}\Conid{A}\;\times\;(\Conid{B}\to \Conid{C})} is
functorial on \ensuremath{\Conid{C}}, with action on morphisms given by \ensuremath{\Conid{R}\;(\Varid{f})\;(\Varid{a},\Varid{g})\mathrel{=}(\Varid{a},\Varid{f}\circ\Varid{g})}. Hence, we can instantiate a polymorphic function \ensuremath{\Varid{h}\colon\Conid{T}} to the
functor \ensuremath{\Conid{R}} and obtain \ensuremath{h_{R}\colon(\Conid{A}\to \Conid{R}\;\Conid{B})\to \Conid{R}\;\Conid{C}}, which amounts to the type
\ensuremath{h_{R}\colon(\Conid{A}\to (\Conid{A}\times(\Conid{B}\to \Conid{B})))\to \Conid{A}\times(\Conid{B}\to \Conid{C})}.

\begin{hscode}\SaveRestoreHook
\column{B}{@{}>{\hspre}l<{\hspost}@{}}%
\column{E}{@{}>{\hspre}l<{\hspost}@{}}%
\>[B]{}\varphi\colon(\forall \Conid{F}\hsforall \colon\FN{Functor}.~(\Conid{A}\to \Conid{F}\;\Conid{B})\to \Conid{F}\;\Conid{C})\to \Conid{A}\times(\Conid{B}\to \Conid{C}){}\<[E]%
\\
\>[B]{}\varphi\;\Varid{h}\mathrel{=}h_{R}\;(\lambda \Varid{a}.~(\Varid{a},id_B)){}\<[E]%
\ColumnHook
\end{hscode}\resethooks
The proof that \ensuremath{\varphi} and \ensuremath{\varphi^{-1}} are indeed inverses will be given for
a $\Set$ model in Section~\ref{sec:representation}.

The simple representation \ensuremath{\Conid{A}\times(\Conid{B}\to \Conid{C})} is possible due to the
restrictive nature of the type $T$: all we know about $F$ is that it
is a functor. What happens when $F$ has more structure?

Consider now the type 
$$T'=\ensuremath{\forall \Conid{F}\hsforall \colon\FN{Pointed}.~(\Conid{A}\to \Conid{F}\;\Conid{B})\to \Conid{F}\;\Conid{C}}.$$
In this case $F$ ranges over \emph{pointed functors}. That is, $F$ is
a functor equipped with a natural transformation $\ensuremath{\eta }_X : \ensuremath{\Conid{X}\to \Conid{F}\;\Conid{X}}$.
An inhabitant of $T'$ is a function $h = \lambda g.\, r$, where $r$ can be
obtained in the same manner as before, or else by applying the point
$\eta_C$ to a given $c\in C$. Hence, a simpler type representing
$T'$ seems to be $(A\times(\ensuremath{\Conid{B}\to \Conid{C}})) + C$.

More formally, we want an isomorphism
$$
\ensuremath{\forall \Conid{F}\hsforall \colon\FN{Pointed}.~(\Conid{A}\to \Conid{F}\;\Conid{B})\to \Conid{F}\;\Conid{C}} 
\xymatrix{~ \ar@/^/@<1ex>[r]^{\varphi'}
               \ar@{}[r]|{\cong}
& 
     \ar@/^/@<1ex>[l]^{\varphi'^{-1}}~} 
(A\times(\ensuremath{\Conid{B}\to \Conid{C}})) + C.
$$
The definition of $\varphi'^{-1}$ is the following.
\begin{hscode}\SaveRestoreHook
\column{B}{@{}>{\hspre}l<{\hspost}@{}}%
\column{8}{@{}>{\hspre}l<{\hspost}@{}}%
\column{20}{@{}>{\hspre}l<{\hspost}@{}}%
\column{31}{@{}>{\hspre}l<{\hspost}@{}}%
\column{E}{@{}>{\hspre}l<{\hspost}@{}}%
\>[B]{}\varphi^{\prime-1}{}\<[8]%
\>[8]{}\colon(\Conid{A}\times(\Conid{B}\to \Conid{C}))\mathbin{+}\Conid{C}{}\<[31]%
\>[31]{}\to \forall \Conid{F}\hsforall \colon\FN{Pointed}.~(\Conid{A}\to \Conid{F}\;\Conid{B})\to \Conid{F}\;\Conid{C}{}\<[E]%
\\
\>[B]{}\varphi^{\prime-1}\;(\textsf{inl}\;(\Varid{a},\Varid{k})){}\<[20]%
\>[20]{}\mathrel{=}\lambda \Varid{g}.~\Conid{F}\;(\Varid{k})\;(\Varid{g}\;\Varid{a}){}\<[E]%
\\
\>[B]{}\varphi^{\prime-1}\;(\textsf{inr}\;\Varid{c}){}\<[20]%
\>[20]{}\mathrel{=}\lambda \anonymous .~\eta_C\;\Varid{c}{}\<[E]%
\ColumnHook
\end{hscode}\resethooks
In order to define \ensuremath{\varphi^\prime}, notice that \ensuremath{\Conid{R'}\;\Conid{C}\mathrel{=}(\Conid{A}\;\times\;(\Conid{B}\to \Conid{C}))\mathbin{+}\Conid{C}}
is a pointed functor on \ensuremath{\Conid{C}}, with \ensuremath{\eta \mathrel{=}\textsf{inr}}.  Hence, we can
instantiate a polymorphic function \ensuremath{\Varid{h}\colon\Conid{T'}} to the pointed functor
\ensuremath{\Conid{R'}} to obtain \ensuremath{h_{R'}\colon(\Conid{A}\to \Conid{R'}\;\Conid{B})\to \Conid{R'}\;\Conid{C}}, or equivalently \ensuremath{h_{R'}\colon(\Conid{A}\to ((\Conid{A}\times(\Conid{B}\to \Conid{B}))\mathbin{+}\Conid{B}))\to (\Conid{A}\times(\Conid{B}\to \Conid{C}))\mathbin{+}\Conid{C}}.

\begin{hscode}\SaveRestoreHook
\column{B}{@{}>{\hspre}l<{\hspost}@{}}%
\column{E}{@{}>{\hspre}l<{\hspost}@{}}%
\>[B]{}\varphi^\prime\colon(\forall \Conid{F}\hsforall \colon\FN{Pointed}.~(\Conid{A}\to \Conid{F}\;\Conid{B})\to \Conid{F}\;\Conid{C})\to (\Conid{A}\times(\Conid{B}\to \Conid{C}))\mathbin{+}\Conid{C}{}\<[E]%
\\
\>[B]{}\varphi^\prime\;\Varid{h}\mathrel{=}h_{R'}\;(\lambda \Varid{a}.~\textsf{inl}\;(\Varid{a},id_B)){}\<[E]%
\ColumnHook
\end{hscode}\resethooks

We can play the same game in the case where the universally quantified
functor is an applicative functor.
$$T''=\ensuremath{\forall \Conid{F}\hsforall \colon\FN{Applicative}.~(\Conid{A}\to \Conid{F}\;\Conid{B})\to \Conid{F}\;\Conid{C}}.$$
An applicative functor is 
a pointed functor $F$ equipped with a multiplication operation
$\ensuremath{\star}_{X,Y} : (F X \times F Y) \to F (X\times Y)$ natural in
$X$ and $Y$, which is coherent with the
point~(a precise definition is given in Section~\ref{sec:applicative}).
%
An inhabitant of $T''$ is a function $h = \lambda g.\, r$, where $r$ can be
obtained by applying the argument $g$ to $n$ elements of $A$ to obtain an $(\ensuremath{\Conid{F}\;\Conid{B}})^n$, then joining the results with the multiplication of the
applicative functor to obtain an $\ensuremath{\Conid{F}\;(B^n)}$, and finally applying a
function \ensuremath{B^n\to \Conid{C}} which takes $n$ elements of $B$ and yields a $C$.
$$
\ensuremath{\forall \Conid{F}\hsforall \colon\FN{Applicative}.~(\Conid{A}\to \Conid{F}\;\Conid{B})\to \Conid{F}\;\Conid{C}} 
\xymatrix{~ \ar@/^/@<1ex>[r]^{\varphi''}
               \ar@{}[r]|{\cong}
& 
     \ar@/^/@<1ex>[l]^{\varphi''^{-1}}~} 
\sum_{n\in\Nat} (A^n\times(B^n \to  C)).
$$
The definition of $\varphi''^{-1}$ is the following.
\begin{hscode}\SaveRestoreHook
\column{B}{@{}>{\hspre}l<{\hspost}@{}}%
\column{9}{@{}>{\hspre}l<{\hspost}@{}}%
\column{19}{@{}>{\hspre}l<{\hspost}@{}}%
\column{37}{@{}>{\hspre}l<{\hspost}@{}}%
\column{E}{@{}>{\hspre}l<{\hspost}@{}}%
\>[B]{}\varphi^{\prime\prime-1}{}\<[9]%
\>[9]{}\colon(\sum_{n\in\Nat}\;(A^n\times(B^n\to \Conid{C}))){}\<[37]%
\>[37]{}\to \forall \Conid{F}\hsforall \colon\FN{Applicative}.~(\Conid{A}\to \Conid{F}\;\Conid{B})\to \Conid{F}\;\Conid{C}{}\<[E]%
\\
\>[B]{}\varphi^{\prime\prime-1}\;(\Varid{n},\Varid{as},\Varid{k}){}\<[19]%
\>[19]{}\mathrel{=}\lambda \Varid{g}.~\Conid{F}\;(\Varid{k})\;(\FN{collect}_n\;\Varid{g}\;\Varid{as}){}\<[E]%
\ColumnHook
\end{hscode}\resethooks

\noindent 
Here, \ensuremath{\FN{collect}_n\colon\forall \Conid{F}\hsforall \colon\FN{Applicative}.~(\Conid{A}\to \Conid{F}\;\Conid{B})\to \Conid{A}^n\to \Conid{F}\;(\Conid{B}^n)} is the function that uses the applicative multiplication to collect
all the applicative effects, i.e. 
$$
\ensuremath{\FN{collect}_n\;\Varid{h}\;(\Varid{x}_{1},\ldots,{x_n})\mathrel{=}\Varid{h}\;\Varid{x}_{1}\star\ldots\star\Varid{h}\;{x_n}}.
$$
In order to define \ensuremath{\varphi^{\prime\prime}}, notice that $\ensuremath{\Conid{R''}\;\Conid{C}} =\sum_{n\in\Nat}
(A^n\times(B^n \to C))$ is an applicative functor on \ensuremath{\Conid{C}}, with $\eta c
= (0,*,\lambda x:1.\ c)$, where $*$ is the sole inhabitant of $1$, and
the multiplication is given by
\begin{hscode}\SaveRestoreHook
\column{B}{@{}>{\hspre}l<{\hspost}@{}}%
\column{E}{@{}>{\hspre}l<{\hspost}@{}}%
\>[B]{}(\Varid{n},\Varid{as},\Varid{k})\star(\Varid{n'},\Varid{as'},\Varid{k'})\mathrel{=}(\Varid{n}\mathbin{+}\Varid{n'},\Varid{as}\plus \Varid{as'},\lambda \Varid{bs}.~(\Varid{k}\;(\FN{take}\;\Varid{n}\;\Varid{bs}),\Varid{k'}\;(\FN{drop}\;\Varid{n}\;\Varid{bs}))){}\<[E]%
\ColumnHook
\end{hscode}\resethooks
Hence, we can instantiate a polymorphic function $h : T''$ to the
applicative functor \ensuremath{\Conid{R''}} to obtain \ensuremath{h_{R''}\colon(\Conid{A}\to \Conid{R''}\;\Conid{B})\to \Conid{R''}\;\Conid{C}}, or
equivalently \ensuremath{h_{R''}\colon(\Conid{A}\to \sum_{n\in\Nat}\;(A^n\times(B^n\to \Conid{B})))\to \sum_{n\in\Nat}\;(A^n\times(B^n\to \Conid{C}))}.

\begin{hscode}\SaveRestoreHook
\column{B}{@{}>{\hspre}l<{\hspost}@{}}%
\column{E}{@{}>{\hspre}l<{\hspost}@{}}%
\>[B]{}\varphi^{\prime\prime}\colon(\forall \Conid{F}\hsforall \colon\FN{Applicative}.~(\Conid{A}\to \Conid{F}\;\Conid{B})\to \Conid{F}\;\Conid{C})\to \sum_{n\in\Nat}\;(A^n\times(B^n\to \Conid{C})){}\<[E]%
\\
\>[B]{}\varphi^{\prime\prime}\;\Varid{h}\mathrel{=}h_{R''}\;(\lambda \Varid{a}.~(\mathrm{1},\Varid{a},id_B)){}\<[E]%
\ColumnHook
\end{hscode}\resethooks

We have seen three different isomorphisms which yield concrete
representations for second-order functionals which quantify over a
certain class of functors (plain functors, pointed functors, and
applicative functors, respectively.)
The construction of each of the three isomorphisms has a similar
structure, so it is natural to ask what the common pattern is.
In order to answer this question and provide a general representation
theorem we will make good use of the power of abstraction of category
theory.

\section{Categorical preliminaries}
\label{sec:catprelim}




%
%
%
%
%

A category $\cC$ is said to be \emph{locally small} when the
collection of morphisms between any two objects $X$ and $Y$ is a
proper set. 
A locally small category is said to be \emph{small} if its collection
of objects is a proper set.
%
We denote by $\Hom{\cC}{X}{Y}$ the (not necessarily small) set of
morphisms between $X$ and $Y$ and extend it to a functor
$\Hom{\cC}{X}{-}$ (the covariant Hom functor). When the category is
$\Set$ (the category of sets and total functions) we will omit the
category from the notation and write $X \tto
Y$. 
Given two categories $\cC$ and $\cD$ we will denote by $\FC{\cC}{\cD}$
the category which has as objects functors $F: \cC\to\cD$ and natural
transformations as morphisms. 
%
A \emph{subcategory} $\cD$ of a category $\cC$ consists of a
collection of objects and morphisms of $\cC$ which is closed under the
operations domain, codomain, composition, and identity. When, for
every object $X$ and $Y$ of $\cD$ subcategory of $\cC$, we have
$\Hom{\cD}{X}{Y} = \Hom{\cC}{X}{Y}$, we say that $\cD$ is a
\emph{full} subcategory of $\cC$.

\subsection{The Yoneda lemma}

The main result of this article hinges on the following famous result:

\begin{theorem}[Yoneda lemma]
Given a locally small category $\cC$, the Yoneda Lemma establishes the
following isomorphism:
\[
  \Hom{\FC{\cC}{\Set}}{(\Hom{\cC}{B}{-})}{F} \quad\cong\quad F\,B
\]%
\noindent natural in object $B : \cC$ and functor $F: \cC \to \Set$.  

That is, the set $F\,B$ is naturally isomorphic to the set of natural
transformations between the functor $(\Hom{\cC}{B}{-})$ and the
functor $F$.

Naturality in $B$ means that given any morphism $h: B \to C$, the
following diagram commutes:
\[
\xymatrix{
(\Hom{\FC{\cC}{\Set}}{(\Hom{\cC}{B}{-})}{F}) \ar[r]^-{\cong} 
  \ar[d]_{\Hom{\FC{\cC}{\Set}}{(\Hom{\cC}{h}{-})}{F}}  &  FB \ar[d]^{Fh} \\
(\Hom{\FC{\cC}{\Set}}{(\Hom{\cC}{C}{-})}{F}) \ar[r]_-{\cong}& FC 
}
\]%
Naturality in $F$ means that given any natural transformation $\alpha:
F \to G$, the following diagram commutes:
\[
\xymatrix{
(\Hom{\FC{\cC}{\Set}}{(\Hom{\cC}{B}{-})}{F}) \ar[r]^-{\cong} 
  \ar[d]_{\Hom{\FC{\cC}{\Set}}{(\Hom{\cC}{B}{-})}{\alpha}}  &  FB \ar[d]^{\alpha_B} \\
(\Hom{\FC{\cC}{\Set}}{(\Hom{\cC}{B}{-})}{G}) \ar[r]_-{\cong}& GB 
}
\]%
The construction of the isomorphism is as follows:
\begin{itemize}
\item Given a natural transformation $\alpha: (\Hom{\cC}{B}{-}) \to F$, its
component at $B$ is a function $\alpha_B: (\Hom{\cC}{B}{B}) \to FB$. Then, the corresponding element of $F\,B$ is $\alpha_B(\ensuremath{id_B})$. 

\item For the other direction, given $x : F\,B$, we construct a natural
transformation $\alpha: (\Hom{\cC}{B}{-}) \to F$ in the following manner:
the component at each object $C$, namely $\alpha_C: (\Hom{\cC}{B}{C}) \to
FC$ is given by $\lambda f : B \to C.\, F(f)(x)$.

\end{itemize}
We leave as an exercise for the reader to check that this
construction indeed yields a natural isomorphism.
\end{theorem}

In order to make the relation between the programs and the category
theory more evident, it is convenient to express the Yoneda lemma in
end form:
\begin{equation}
  \label{eq:yoneda_end}
\int_{X\in\cC}    (\Hom{\cC}{B}{X}) \tto F\,X  \quad\cong \quad F\,B  
\end{equation}
The intuition is that an end corresponds to a universal quantification
in a programming language~\cite{BainbridgeFSS90}, and therefore the
above isomorphism could be understood as stating an isomorphism of
types:
\[
  \ensuremath{\forall \Conid{X}\hsforall .~(\Conid{B}\to \Conid{X})\to \Conid{F}\;\Conid{X}} \quad\cong\quad FB
\]%
 Hence, functional programmers not used to
categorical ends can get the intuitive meaning just by replacing in
their minds ends by universal quantifiers. 
The complete definition of end can be found in
Appendix~\ref{sec:ends}. More details can be found in the standard
reference~\cite{macLaneS:catwm}.

%

A simple application of the Yoneda lemma which will be used in the next
section is the following proposition.

\begin{prop}\label{prop:yoneda1}
  Consider an endofunctor $F: \Set \to \Set$, and the functor $R : \Set\times\op\Set\times\Set\to\Set$ defined as $R\,(A,B,X) = A\times(B \tto X)$, 
$R\,(f,g,h)(a,x)= (f a, g \circ x \circ h)$, 
where we write $R_{A,B}X$ for $R\,(A,B,X)$. Then 
  \begin{equation}
    \label{eq:atofb}
A \tto F\,B ~\cong~  \Hom{\FC{\Set}{\Set}}{R_{A,B}}{F}
  \end{equation}
\end{prop}

\begin{proof}
 
  \begin{calculation}
 \begin{array}[b]{cl}
   & A \tto F\,B \\
\cong & \comment{Yoneda  }\\
   &  A \tto \int_{X}   ( (B \tto X) ~\tto~ F\,X) \\
\cong & \comment{Hom functors preserve ends (Remark~\ref{remark:ends_as_limits})} \\
      &  \int_{X}  A \tto  ((B \tto X) ~\tto~ F\,X) \\
\cong & \comment{Adjoints (currying)} \\
      &  \int_{X}  A\times(B \tto X) ~\tto~ F\,X \\
\cong & \comment{Definition of \ensuremath{R_{A,B}}} \\
      &  \int_{X}  R_{A,B}\,X ~\tto~ F\,X \\
\cong & \comment{Natural transformations as ends} \\
      &  \Hom{\FC{\Set}{\Set}}{R_{A,B}}{F}
  \end{array}
\end{calculation}

More concretely, the isomorphism is witnessed by the following functions:
$$
\begin{array}{l@{\,}c@{\,}l}
\alpha_F & : &  \ensuremath{(\Conid{A}\to \Conid{F}\;\Conid{B})} \to  \Hom{\FC{\Set}{\Set}}{R_{A,B}}{F} \\
\alpha_F(f) & = & \tau\quad\text{where~} 
                    \begin{array}[t]{l@{\,}c@{\,}l}
                       \tau_X & : & A\times(B\tto X) \to F\,X \\
                       \tau_X (a,g) & = & F(g)(f(a)) \\
                     \end{array}
\\
\alpha_F^{-1} & : &   (\Hom{\FC{\Set}{\Set}}{R_{A,B}}{F}) \to \ensuremath{\Conid{A}\to \Conid{F}\;\Conid{B}} \\
\alpha_F^{-1}(h) & = & \lambda a.~ h_B\,(a,\ensuremath{\Varid{id}}_B)
\end{array}
$$%
This isomorphism is natural in $A$ and $B$.
\end{proof}

\subsection{Adjunctions}

An adjunction is a relation between two categories which is weaker
than isomorphism of categories.
\begin{defn}[Adjunction]
Given categories $\cC$ and $\cD$, functors $L: \cC \to
\cD$ and $R: \cD \to \cC$, an \emph{adjunction} is given by a tuple
$(L,R,\ladj{-},\radj{-})$, where $\ladj{-}$ and $\radj{-}$ are the
components of the following isomorphism:
\begin{equation}\label{eq:adjunction-iso}
      \ladj{-} : \Hom{\cD}{L\,C}{D} 
                 \quad \cong \quad 
                 \Hom{\cC}{C}{R\,D} : \radj{-}  
\end{equation}
\noindent which is natural in $C\in\cC$ and $D\in\cD$.
That is, for $f : L\,C \to D$ and $g : C \to R\,D$ we have
\begin{equation}
  \label{eq:adjunction-equiv}
  \ladj{f} = g \quad\Leftrightarrow\quad f = \radj{g}
\end{equation}
The components of the isomorphism $\ladj{-}$ and $\radj{-}$ are called
\emph{adjuncts}.
That the isomorphism is natural means that for any $C,C'\in\cC$;
$D,D'\in\cD$; $h : C'\to C$; $k : D \to D'$; $f : L\,C \to D$; and $g :
C \to R\,D$, the following equations hold:
\begin{eqnarray}
  \label{eq:adjunction-naturality}
R\,k \circ \ladj{f} \circ h & = & \ladj{k \circ f \circ L\,h}
\label{eq:adj-nat1}
\\
k \circ \radj{g} \circ L\,h & = & \radj{R\,k\circ g \circ h}  
\label{eq:adj-nat2}
\end{eqnarray}

We indicate the categories involved in an adjunction by writing $\cC
\rightharpoonup \cD$ (note the asymmetry in the notation), and often
leave the components of the isomorphism implicit and simply write
$L\dashv R$.

\end{defn}

The \emph{unit} $\eta$ and \emph{counit} $\varepsilon$ of the
adjunction are defined as:
\begin{equation}
  \label{eq:eta-epsilon-from-adjuncts}
  \eta = \ladj{id} \qquad\qquad \varepsilon= \radj{id};
\end{equation}
The adjuncts can be characterised in terms of the unit and counit:
\begin{equation}
  \label{eq:adjuncts-from-unit-counit}
  \ladj{f} = R\,f\circ\eta
  \qquad\qquad
  \radj{g} = \varepsilon \circ L\,g.
\end{equation}
 For more details, see~\cite{macLaneS:catwm,awodey2006ct}.


%


\section{A representation theorem for second-order functionals}
\label{sec:representation}

Consider a small subcategory $\cat{F}$ of $\FC{\Set}{\Set}$, the
category of endofunctors on $\Set$.\footnote{ We are interested in
  functors representable in a programming language, such as realisable
  functors~\cite{BainbridgeFSS90,ReynoldsP93}. Therefore, it is
  reasonable to assume smallness.}  By Yoneda,

\newcommand{\F}{F}

  \begin{equation}
    \label{eq:hoyoneda}
   \int_{\F\in\cat{F}} (\Hom{\cat{F}}{G}{\F}) \tto H\,\F \quad\cong\quad H\,G    
  \end{equation}
%
  Note that $G$ is any functor in $\cat{F}$ and $H$ is any functor
  $\cat{F}\to\Set$. In particular, given a set $X$, we obtain the
  functor $(-X) : \cat{F}\to\Set$ that applies a functor in $\cat{F}$
  to $X$.
  That is, the action on objects is $F \mapsto F\,X$. The above
  equation, specialised to $(-X)$ is
\begin{equation}
  \label{eq:hoyonedax}
\forall G\in\cat{F}.\qquad
  \int_\F (\Hom{\cat{F}}{G}{\F}) \tto \F\,X \quad\cong\quad G\,X 
\end{equation}

For example, let $R_{A,B}\,X=A\times(B \tto X)$ as in
Proposition~\ref{prop:yoneda1}, and let $\cat{E}$ be a small
full sub-category of $\FC{\Set}{\Set}$ such that $R_{A,B} \in \cat{E}$.

  Then, we calculate
\begin{calculation}
\begin{array}{cl}
   &  \int_{\F\in\cat{E}} (A \tto \F\, B) \tto \F\, X
   \\
\cong & \comment{Equation~(\ref{eq:atofb})} \\
   &   \int_{\F\in\cat{E}} (\Hom{\cat{E}}{R_{A,B}}{\F}) \tto \F\, X
   \\
\cong & \comment{Equation (\ref{eq:hoyonedax})} \\
      &  R_{A,B}\,X
  \end{array}
\end{calculation}

That is, we have proven that 
\begin{equation}\label{eq:simpleRepresentation}
  \int_{\F} (A \tto \F\,B) \tto \F\,X ~\cong~ R_{A,B}\,X
\end{equation}
This isomorphism provides a justification for the first isomorphism of
the introduction, namely:
\[
  \ensuremath{\forall \Conid{F}\hsforall \colon\FN{Functor}.~(\Conid{A}\to \Conid{F}\;\Conid{B})\to \Conid{F}\;\Conid{C}}\quad\cong\quad \ensuremath{\Conid{A}\times(\Conid{B}\to \Conid{C})}
\]%


\subsection{Unary representation theorem}

Let us now consider categories of endofunctors that carry some structure.
For example, a category $\cat{F}$ may be the category of monads and
monad morphisms, or the category of applicative functors and
applicative morphisms. Then we have a functor that forgets the extra
structure and yields a plain functor. For example, the forgetful
functor $U: \Mon \to \cat{E}$ maps a monad $(T,\mu,\eta)\in \Mon$ to
the endofunctor $T$, forgetting that the functor has a monad structure
given by $\mu$ and $\eta$.
It often happens that this forgetful functor has a left adjoint $(-)^*
: \cat{E}\to\cat{F}$. Such an adjoint takes an arbitrary endofunctor
$F$ and constructs the \emph{free} structure on $F$. For example, in
the monad case, $F^*$ would be the free monad on $F$. The adjunction
establishes the following natural isomorphism between morphisms in
$\cat{F}$ and $\cat{E}$:
\begin{equation}
  \label{eq:adjoint}
   \Hom{\cat{F}}{E^*}{F} \quad\cong\quad \Hom{\cat{E}}{E}{UF}
\end{equation}
In this situation we have the following representation theorem.

\begin{theorem}[Unary representation]\label{thm:representation1}
  Consider an adjunction $((-)^*,U,\ladj{-},\radj{-}) :
  \cat{E}\rightharpoonup\cat{F}$, where $\cat{F}$ is
  small and $\cat{E}$ is a full
  subcategory of $\FC{\Set}{\Set}$ such that the family of functors
  $R_{A,B}\,X= A\times (B \tto X)$ is in $\cat{E}$.  Then, we have the following isomorphism natural in $A$, $B$,
  and $X$.
\begin{equation}
  \label{eq:representation}
 \int_{\F}  (A \tto U\!\F\,B) \tto U\!\F\,X \quad\cong\quad UR_{A,B}^*\,X 
\end{equation}
\end{theorem}

\begin{proof}

  \begin{calculation}
    \begin{array}{cl}
      &   \int_{\F}  (A \tto U\!\F\, B) \tto U\!\F\, X 
      \\
      \cong & \comment{Equation~(\ref{eq:atofb})} \\
      &   \int_{\F} (\Hom{\cat{E}}{R_{A,B}}{U\!\F}) \tto U\!\F\, X
      \\
      \cong & \comment{ \ensuremath{(-)^*} is left adjoint to \ensuremath{U} (see Eq.~\ref{eq:adjoint})} \\
      &   \int_{\F} (\Hom{\cat{F}}{R_{A,B}^*}{\F}) \tto U\!\F\, X
      \\
      \cong & \comment{Yoneda
      } \\
      &  UR_{A,B}^*\,X
    \end{array}
  \end{calculation}

Every isomorphism in the proof is natural in $X$, the first one is
natural in $A$ and $B$, and the last two are natural in
$R_{A,B}$. Therefore, the resulting isomorphism is also natural in $A$
and $B$.
%
\end{proof}

Since the free pointed functor on $F$ is simply $F^* = F + Id$, and
the free applicative functor on small functors such as $R_{A,B}$
exists~\cite{Capriotti2014}, this theorem explains all the
isomorphisms in the introduction. Furthermore, it explains the
structure of the representation functor (it is the free construction
on $R_{A,B}$) and what's more, it tells us that the isomorphism is
natural. 

For the sake of concreteness, we present the functions witnessing the
isomorphism in the theorem:
\[
\begin{array}{l@{\,}c@{\,}l}
\varphi & : & (\int_{\F}  (A \tto U\!\F\, B) \tto U\!\F\, X) \to UR_{A,B}^*\,X  \\
\varphi(h) & = & h_{R_{A,B}^*}\,(\alpha^{-1}_{UR_{A,B}^*}(\eta_{R_{A,B}}))
\\
\\
\varphi^{-1} & : &  UR_{A,B}^*\,X \to \int_{\F}  (A \tto U\!\F\, B) \tto U\!\F\, X \\
\varphi^{-1}(r) & = & \tau \quad\text{where~}
\begin{array}[t]{l@{\,}c@{\,}l}
\tau_{\F} & : & (A \tto U\!\F\, B) \tto U\!\F\, X \\
\tau_{\F}(g) & = & (U\,\radj{\alpha_{UF}(g)}_X)(r)  
\end{array} \\
\end{array}
\]%
Here, $\eta$ is the unit of the
adjunction, and $\alpha$ is the isomorphism in
Proposition~\ref{prop:yoneda1}.

\subsection{Generalisation to many functional arguments}

Let us consider functionals of the form
\[
\ensuremath{\forall \Conid{F}\hsforall .~(\Conid{A}_{1}\to \Conid{F}\;\Conid{B}_{1})\to \ldots\to (A_n\to \Conid{F}\;B_n)\to \Conid{F}\;\Conid{X}}
\]%
The representation theorem, Theorem~\ref{thm:representation1}, can be
easily generalised to include the above functional.

\begin{theorem}[N-ary representation]
\label{thm:representation2}
Consider an adjunction $((-)^*,U,\ladj{-},\radj{-}) :
\cat{E}\rightharpoonup\cat{F}$, where $\cat{F}$ is small and $\cat{E}$
is a full subcategory of $\FC{\Set}{\Set}$ closed under coproducts
such that the family of functors $R_{A,B}\,X= A\times (B \tto X)$ is
in $\cat{E}$.
Let
  $A_i, B_i$ be sets for $i\in\{1,\dots,n\}, n\in\Nat$. 
  Then, we have the following isomorphism
\begin{equation}
  \label{eq:n-ary_representation}
 \int_{\F}  (\prod_i (A_i \tto U\!\F\,B_i)) \tto U\!\F\,X 
   \quad\cong\quad 
   U(\sum_i R_{A_i,B_i})^*\,X 
\end{equation}
\noindent natural in $A_i$, $B_i$, and $X$.
\end{theorem}

\begin{proof}
  The proof follows the same path as the one in
  Theorem~\ref{thm:representation1}, except that now we use the
  isomorphism $\ensuremath{(\Conid{A}\to \Conid{C})\times(\Conid{B}\to \Conid{C})\;\cong\;(\Conid{A}\mathbin{+}\Conid{B})\to \Conid{C}}$ that results
  from the universal property of coproducts.
More precisely, the proof is as follows:

\begin{calculation}
  \begin{array}{cl}
    &   \int_{\F}  (\prod_i (A_i \tto U\!\F\, B_i)) \tto U\!\F\, X 
    \\
    \cong & \comment{Equation~(\ref{eq:atofb})} \\
    &   \int_{\F} (\prod_i (\Hom{\cat{E}}{R_{A_i,B_i}}{U\!\F})) \tto U\!\F\, X
    \\
    \cong & \comment{Coproducts} \\
    &   \int_{\F} (\Hom{\cat{E}}{\sum_i R_{A_i,B_i}}{U\!\F}) \tto U\!\F\, X
    \\
    \cong & \comment{ \ensuremath{(-)^*} is left adjoint to \ensuremath{U} (see Eq.~\ref{eq:adjoint})} \\
    &   \int_{\F} (\Hom{\cat{F}}{(\sum_i R_{A_i,B_i})^*}{\F}) \tto U\!\F\, X
    \\
    \cong & \comment{Yoneda } \\
    &  U(\sum_i R_{A_i,B_i})^*\,X
  \end{array}
\end{calculation}

 Naturality follows from naturality of its component isomorphisms. 
\end{proof}



%


\section{Parameterised comonads and very well-behaved lenses}
\label{sec:pcomonads_and_lenses}

The functor $R_{A,B}\,X = A\times (B\tto X)$ plays a fundamental role
in Theorems~\ref{thm:representation1} and~\ref{thm:representation2}.
Such a functor $R$ has the structure of a parameterised
comonad~\cite{paramnotions-jfp,algebras-param-monads} and is sometimes
called a parameterised store comonad.  As a first
application of the representation theorem we analyse the relation
between coalgebras for this parameterised comonad and very
well-behaved lenses~\cite{foster:2007}.

\begin{defn}[Parameterised comonad]
  Fix a category $\cP$ of parameters.  A $\cP$-parameterised comonad
  on a category $\cC$ is a triple $(C,\varepsilon,\delta)$, where:
  \begin{itemize}
  \item $C$ is a functor $\cP\times\op\cP\times\cC \to \cC$. We write
    the parameters as (usually lowercase) subindexes. That is, $C_{a,b}\, X =
    C(a,b,X)$.

  \item the \emph{counit} $\varepsilon$ is a family of morphisms
    $\varepsilon_{a,X}: C_{a,a}\,X \to X$ which is natural in $X$ and
    dinatural in $a$ (dinaturality is defined in Appendix~\ref{sec:ends},
    Definition~\ref{def:dinatural}),
  \item the \emph{comultiplication} $\delta$ is a family of morphisms
  $\delta_{a,b,c,X}: C_{a,c}\,X \to C_{a,b}\,(C_{b,c}\,X)$
  natural in $a, c$ and $X$ and dinatural in $b$.
  \end{itemize}
  These must make the following diagrams commute:

\[
\xymatrix@C+=2cm{
   & C_{a,b}\,X 
       \ar[dl]_{\delta_{a,b,b,X}}
       \ar@{=}[d]
       \ar[dr]^{\delta_{a,a,b,X}}
\\
     C_{a,b}\,(C_{b,b}\,X)
      \ar[r]_-{C_{a,b}\,\varepsilon_{b,X}}
   & C_{a,b}\,X
   & C_{a,a}\,(C_{a,b}\,X)
      \ar[l]^-{\varepsilon_{a,C_{a,b}\,X}}
}
\]%
\[
\xymatrix@C+=3cm{
C_{a,d}\,X   
        \ar[r]^-{\delta_{a,b,d,X}}
        \ar[d]_{\delta_{a,c,d,X}}
  &
 C_{a,b}\,(C_{b,d}\,X)
       \ar[d]^{C_{a,b}\,\delta_{b,c,d,X}}
 \\
 C_{a,c}\,(C_{c,d}\,X)
       \ar[r]_-{\delta_{a,b,c,C_{c,d}\,X}}
 &
 C_{a,b}\,(C_{b,c}\,(C_{c,d}\,X))
}
\]%
\end{defn}

\bigskip
\begin{defn}[Coalgebra for a parameterised comonad]
\label{defn:coalgebra_parameterised_comonad}
  Let $C$ be a $\cP$-parameterised comonad on $\cC$.  Then a
  $C$-coalgebra is a pair $(J,k)$ of a functor $J: \cP\to\cC$, and
  a family $k_{a,b} : J\,a \to C_{a,b}\,(J\,b)$, natural in
  $a$ and dinatural in $b$, such that the following diagrams
  commute:
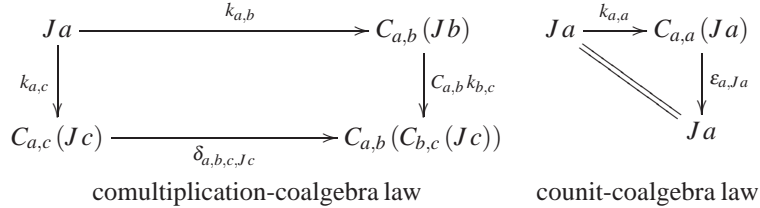
\[
\begin{array}{ccc}
\xymatrix@C+=3cm{
  J\,a 
      \ar[r]^{k_{a,b}}
      \ar[d]_{k_{a,c}}
  &
  C_{a,b}\,(J\,b)
      \ar[d]^-{C_{a,b}\,k_{b,c}}
  \\
  C_{a,c}\,(J \,c)
      \ar[r]_-{\delta_{a,b,c,J\,c}}
  &
  C_{a,b}\,(C_{b,c}\,(J\,c))
}
&
\xymatrix{
  J\,a
    \ar[r]^-{k_{a,a}}
    \ar@{=}[dr]
  &
  C_{a,a}\,(J\,a)
    \ar[d]^{\varepsilon_{a,J\,a}}
  \\
  & J\,a
}
\\
\text{comultiplication-coalgebra law}& \text{counit-coalgebra law}
\end{array}
\]%
\end{defn}

The definitions of parameterised comonad and of coalgebra for a
parameterised comonad are dualisations of the ones for monads found
in~\cite{algebras-param-monads}.

\begin{example}
\label{ex:R}
The functor $R_{a,b}\,X = a \times (b \tto X)$ is a parameterised
comonad, with the following counit and comultiplication:
\[
\begin{array}{lcl}
\varepsilon_{a,X}        &: & R_{a,a}\,X \to X \\
\varepsilon_{a,X}\ (x,f) &= & f x \\
\\
\delta_{a,b,c,X}          &: & R_{a,c}\,X \to R_{a,b}\,(R_{b,c}\,X) \\
\delta_{a,b,c,X}\ (x,f)   &= & (x, \lambda y.\, (y, f))
\end{array}
\]%
  
\end{example}

\begin{example}
\label{example:RK}
  Given a functor $K : \cP \to Set$, define the functor
  $R^{(K)}_{a,b}\,X = Ka \times (Kb \tto  X) : \cP \times \op\cP \times Set \to
  Set$.  For every functor $K$, $R^{(K)}$ is a parameterised comonad, with the
  following counit and comultiplication:
\[
\begin{array}{lcl}
\varepsilon_{a,X}        &: & R^{(K)}_{a,a}\,X \to X \\
\varepsilon_{a,X}\ (x,f) &= & f x \\
\\
\delta_{a,b,c,X}          &: & R^{(K)}_{a,c}\,X \to R^{(K)}_{a,b}\,(R^{(K)}_{b,c}\,X) \\
\delta_{a,b,c,X}\ (x,f)   &= & (x, \lambda y.\, (y, f))
\end{array}
\]%
  
The parameterised comonad $R$ from Example~\ref{ex:R} is the same as $R^{(\Id)}$ where $\Id$ is the
identity functor.
\end{example}

The proposition below shows how the comonadic structure of $R^{(K)}$
interacts nicely with the isomorphism of
Proposition~\ref{prop:yoneda1}.

\begin{prop}\label{prop:epsilon-alpha}
\label{prop:delta-alpha}
Let $F, G : Set \to Set$, $f : a \to F b$, and $g : b \to
G c$, then the following equations hold.
\begin{enumerate}
\item[a)] $\varepsilon_{a,X} = \alpha_{\Id}(id_{K a})_X \quad: R^{(K)}_{a,a,X} \to X$
\item[b)] 
$  (\alpha_{F}(f) \cdot \alpha_{G}(g))_X \circ \delta_{a,b,c,X} = 
   \alpha_{F \cdot G}(Fg \circ f)_X \quad : R^{(K)}_{a,c}\,X \to F(G\,X)$
\end{enumerate}
\noindent where $F \cdot G$ is functor composition and where
$\alpha \cdot \beta$ is the horizontal composition of natural
transformations. That is, given natural transformations $\alpha : F
\to G$, and $\beta : F' \to G'$, horizontal composition $\alpha \cdot
\beta : F\cdot F' \to G \cdot G'$ is given by $\alpha \cdot \beta =
G(\beta) \circ \alpha_{F'}$.
\end{prop}

\begin{example}\label{example:pi1-lens}
The pair $((\times C), k)$ is an $R$-coalgebra with

\[
\begin{array}{lcl}
k_{a,b}  & : & a \times C \to R_{a,b}(b \times C) \\
k_{a,b}\ (a,c) & = & (a, \lambda b. \, (b,c))
\end{array}
\]%
\end{example}

Coalgebras of $R^{(K)}$ play an important role in functional
programming as they are precisely the type of very well-behaved lenses,
hereafter called lenses~\cite{foster:2007}.
A lens provides access to a component $B$ inside another type $A$.
More formally a lens from $A$ to $B$ is an isomorphism $A \cong B \times C$
for some residual type $C$.
A lens from $A$ to $B$ is most easily implemented by a pair of appropriately
typed getter and setter functions
\[
\begin{array}{lcl}
get  & : & A \to B \\
set  & : & A \times B \to A \\
\end{array}
\]%
satisfying three laws\footnote{In Foster et al.~\shortcite{foster:2007},
the less well-behaved lenses do not satisfy all three laws.}
\begin{eqnarray*}
set(x,get(x)) & = & x \\
get(set(x,y)) & = &  y \\
set(set(x,y_1),y_2) & = & set(x,y_2)
\end{eqnarray*}
More generally, given two functors $J : \cP \to Set$ and $K : \cP \to Set$,
we can form a parameterised lens from $J$ to $K$ with a family of
getters and setters
\[
\begin{array}{lcl}
get_{a}  & : & Ja \to Ka \\
set_{a,b}  & : & Ja \times Kb \to Jb \\
\end{array}
\]%
\noindent{}satisfying the same three laws, and with $get$ being natural in $a$
and $set$ being natural in $b$.
By some simple algebra we see that the type of lenses is isomorphic to the type
of coalgebras of the parameterised comonad $R^{(K)}$.
\[
(Ja \tto Ka) \times (Ja \times Kb \tto Jb) \quad\cong\quad Ja \tto R^{(K)}_{a,b}(Jb)
\]%
Furthermore the coalgebra laws are satisfied if and only if the
corresponding lens laws are satisfied~\cite{oconnor:2010,colens}.
%
%
For instance, the coalgebra given in
Example~\ref{example:pi1-lens} is a parameterised lens into the first
component of a pair.

Using the representation theorem and some simple manipulations we can
define a third way to represent a parameterised lens from $J$ to $K$.
The so-called Van Laarhoven
representation~\cite{vanLaarhoven:2009c,oconnor:2011} is defined by
a family of ends
\[
\int_{\F:\cat{E}} (Ka \tto \F (Kb)) \tto Ja \tto \F (Jb)
\]%
that is natural in the sense that given two
arrows from $\cat{P}$, $p : a \to a'$ and $q : b \to b'$, and
given $f : Ka' \tto \F (Kb)$ for some $\F:\cat{E}$ then
\[
  F(Jq) \circ v_{a',b,F}(f) \circ Jp = v_{a,b',F}(F(Kq) \circ f \circ Kp).
\]%

The corresponding laws for the Van Laarhoven representation of lenses
are
\begin{itemize}
\item the linearity law 

For all $f : Ka \tto \F (Kb)$ and $g : Kb \tto G (Kc)$,
\[
  v_{a,c,F \cdot G}(\F g \circ f) = \F v_{b,c,G}(g) \circ v_{a,b,F}(f)
\]%
\item and the unity law
\[
  v_{a,a,\Id}(id_{Ka}) = id_{Ja}.
\]%
\end{itemize}

The following theorem proves that the coalgebra representation and Van
Laarhoven representation of parameterised lenses are equivalent.

\begin{theorem}[Lens representation]\label{theorem:lens-representation}
  Given $\cat{E}$, a small full subcategory of $Set^{Set}$ and
  given functors $J,K : \cP \to Set$, then the families
  $k_{a,b} : Ja \to R^{(K)}_{a,b}(Jb)$ which form $R^{(K)}$-coalgebras
  $(J, k)$ are isomorphic to the families of ends
\[
\int_{\F:\cat{E}} (Ka \tto \F (Kb)) \tto Ja \tto \F (Jb)
\]%
which satisfy the linearity and unity laws.

\end{theorem}
\begin{proof}
  First, we prove the isomorphism of families without regard to the laws

  \begin{calculation}
    \begin{array}{cl}
      & Ja \tto R^{(K)}_{a,b}(Jb) \\
      \cong & \comment{definition of $R^{(K)}$}\\
      & Ja \tto R_{Ka,Kb}(Jb) \\
      \cong & \comment{Equation~\ref{eq:simpleRepresentation}}\\
      & Ja \tto \int_{\F} (Ka \tto \F (Kb)) \tto \F (Jb) \\
      \cong & \comment{Hom functors preserve ends (Remark~\ref{remark:ends_as_limits})}\\
      & \int_{\F} Ja \tto (Ka \tto \F (Kb)) \tto \F (Jb) \\
      \cong & \comment{Swap argument}\\
      & \int_{\F} (Ka \tto \F (Kb)) \tto Ja \tto \F (Jb) \\
    \end{array}
  \end{calculation}

  This isomorphism is witnessed by the following functions:

\[
\begin{array}{l@{\,}c@{\,}l}
\gamma & : & (\int_{\F}  (Ka \tto \F (Kb)) \tto Ja \tto \F (Jb)) \to (Ja \tto R^{(K)}_{a,b}\,(Jb))
\\
\gamma(h) & = & h_{R^{(K)}_{a,b}}\,(\alpha^{-1}_{R^{(K)}_{a,b}}(id))
\\
\\
\gamma^{-1} & : &  (Ja \tto R^{(K)}_{a,b}\,(Jb)) \to \int_{\F} (Ka \tto \F (Kb)) \tto (Ja \tto \F (Jb))\\
\gamma^{-1}(k) & = & \tau \quad\text{where~}
\begin{array}[t]{l@{\,}c@{\,}l}
\tau_{\F} & : & (Ka \tto \F (Kb)) \tto Ja \tto \F (Jb) \\
\tau_{\F}(g) & = & \alpha_{\F}(g)_{Jb} \circ k
\end{array} \\
\end{array}
\]%

In order to prove that the laws of coalgebras for parameterised
comonads correspond to unity and linearity, we first prove two
technical lemmas.
\begin{lemma}\label{lemma:tech1}%
$$
  \gamma^{-1}(k_{a,c})_{F \cdot G}(Fg \circ f)
 = (\alpha_{F}(f) \cdot \alpha_{G}(g))_{Jc} \circ \delta_{a,b,c,Jc} \circ k_{a,c}
$$%
\end{lemma}
\begin{proof}
This follows from Proposition~\ref{prop:delta-alpha}(b).
\end{proof}

\begin{lemma}\label{lemma:tech2}
$$
  F(\gamma^{-1}(k_{b,c})_{G}(g)) \circ \gamma^{-1}(k_{a,b})_{F}(f)
 = (\alpha_{F}(f) \cdot \alpha_{G}(g))_{Jc} \circ R^{(K)}_{a,b}(k_{b,c}) \circ k_{a,b}
$$
\end{lemma}
\begin{proof}
This follows from the definition of $\gamma^{-1}$ and properties of functors and
natural transformations.
\end{proof}


Generalised versions of Lemma~\ref{lemma:tech1} and Lemma~\ref{lemma:tech2}
appear with detailed proofs in
Appendix~\ref{thm:generalised-lens-representation},
Lemma~\ref{lemma:gen-tech1} and Lemma~\ref{lemma:gen-tech2}.

By the previous two lemmas, to prove that the
comultiplication-coalgebra law is equivalent to the linearity law it
suffices to prove the following:

\[
\begin{array}{rcl}
 R^{(K)}_{a,b}(k_{b,c}) \circ k_{a,b} & = & \delta_{a,b,c,Jc} \circ k_{a,c} \\
 & \iff &  \\
 \forall F, G, f, g. (\alpha_{F}(f) \cdot \alpha_{G}(g)) \circ
R^{(K)}_{a,b}(k_{b,c}) \circ k_{a,b} & = &
 (\alpha_{F}(f) \cdot \alpha_{G}(g)) \circ \delta_{a,b,c,Jc} \circ k_{a,c}
\end{array}
\]%
The forward implication is clear.  To prove the reverse implication
take $F = R^{(K)}_{a,b}$ and $f =
\alpha^{-1}_{R^{(K)}_{a,b}}(id)_{Jb}$.  Also take $G = R^{(K)}_{b,c}$
and $g = \alpha^{-1}_{R^{(K)}_{b,c}}(id)_{Jc}$.  Then $\alpha_{F}(f) =
id$ and $\alpha_{G}(g) = id$.  Therefore, $\alpha_{F}(f) \cdot
\alpha_{G}(g) = id$ and the result follows.

To prove that the counit-coalgebra law is equivalent to the unity law
it suffices to prove that $\varepsilon_{a, Ja} \circ k_{a,a} =
\gamma^{-1}(k_{a,a})_{\Id}(id)$.

\begin{calculation}
  \begin{array}[b]{cl}
    & \gamma^{-1}(k_{a,a})_{\Id}(id) \\
    = & \comment{definition of $\gamma^{-1}$}\\
    & \alpha_{\Id}(id)_{Ja} \circ k_{a,a} \\
    = & \comment{Proposition~\ref{prop:epsilon-alpha}(a)}\\
    & \varepsilon_{a, Ja} \circ k_{a,a} \\
  \end{array}
\end{calculation}
\end{proof}

The previous theorem can be generalised to the case where we have an
adjunction.

\begin{theorem}[Generalised lens representation]\label{theorem:generalized-lens-representation}
Let $\cat{E}$ and $\cat{F}$ be two small categories of $\Set$-endofunctors, such that $\cat{E}$ and $\cat{F}$ are (strict)
monoidal with respect to the identity functor $\Id$ and functor
composition $-\cdot-$, and $\cat{E}$ is a full sub-category.  Let $(-)^* \dashv U : \cat{E}\rightharpoonup \cat{F}$,
be an adjunction between them, such that $U$ is strict monoidal.
Then
\begin{enumerate}
\item $UR^{(K)*}$ is a parameterised comonad.
\item 
  Given functors $J,K : \cP \to Set$, then the family
  $k_{a,b} : Ja \tto UR^{(K)*}_{a,b}(Jb)$ which form the $UR^{(K)*}$-coalgebras
  $(J, k)$ are isomorphic to the family of ends
\[
\int_{\F:\cat{F}} (Ka \tto U\F (Kb)) \tto Ja \tto U\F (Jb)
\]%
which satisfy the linearity and unity laws.
\end{enumerate}
\end{theorem}

\begin{proof}
See Appendix~\ref{thm:generalised-lens-representation},
Proposition~\ref{prop:glr}.
\end{proof}

By considering the identity adjunction between $\cat{E}$ and itself,
Theorem~\ref{theorem:lens-representation} can be recovered from this
generalised version.

\subsection{Implementing lenses in Haskell}\label{sec:lenses_in_Haskell}

The Lens representation theorem demonstrates that the coalgebra
representation of lenses and the Van Laarhoven representation are
isomorphic.  Both representations can be implemented in Haskell.

\begin{hscode}\SaveRestoreHook
\column{B}{@{}>{\hspre}l<{\hspost}@{}}%
\column{3}{@{}>{\hspre}l<{\hspost}@{}}%
\column{E}{@{}>{\hspre}l<{\hspost}@{}}%
\>[3]{}\mbox{\onelinecomment  Parameterised store comonad}{}\<[E]%
\\
\>[3]{}\mathbf{data}\;\Conid{PStore}\;\Varid{a}\;\Varid{b}\;\Varid{x}\mathrel{=}\Conid{PStore}\;(\Varid{b}\to \Varid{x})\;\Varid{a}{}\<[E]%
\\[\blanklineskip]%
\>[3]{}\mbox{\onelinecomment  Coalgebra representation of lenses}{}\<[E]%
\\
\>[3]{}\mathbf{newtype}\;\Conid{KLens}\;\Varid{ja}\;\Varid{jb}\;\Varid{ka}\;\Varid{kb}\mathrel{=}\Conid{KLens}\;(\Varid{ja}\to \Conid{PStore}\;\Varid{ka}\;\Varid{kb}\;\Varid{jb}){}\<[E]%
\\[\blanklineskip]%
\>[3]{}\mbox{\onelinecomment  Van Laarhoven representation of lenses}{}\<[E]%
\\
\>[3]{}\mathbf{type}\;\Conid{VLens}\;\Varid{ja}\;\Varid{jb}\;\Varid{ka}\;\Varid{kb}\mathrel{=}\forall \Varid{f}\hsforall .~\FN{Functor}\;\Varid{f}\Rightarrow (\Varid{ka}\to \Varid{f}\;\Varid{kb})\to \Varid{ja}\to \Varid{f}\;\Varid{jb}{}\<[E]%
\ColumnHook
\end{hscode}\resethooks

There are a few observations to make about this Haskell code.
Firstly, neither the coalgebra laws nor the linearity and unity laws
of the Van Laarhoven representation can be enforced by Haskell's type
system, as it often happens when implementing algebraic structures
such as monoids or monads.
We have accordingly omitted writing out the parameterised comonad operations of
\ensuremath{\Conid{PStore}}.  Secondly, rather than taking $J$ and $K$ as parameters, we
take source and target types for each functor. By not
explicitly using functors as parameters, we avoid
\ensuremath{\mathbf{newtype}} wrapping and unwrapping functions that would otherwise be
needed.  Consider the example of building a lens to access the first component
of a pair.

\begin{hscode}\SaveRestoreHook
\column{B}{@{}>{\hspre}l<{\hspost}@{}}%
\column{3}{@{}>{\hspre}l<{\hspost}@{}}%
\column{21}{@{}>{\hspre}l<{\hspost}@{}}%
\column{E}{@{}>{\hspre}l<{\hspost}@{}}%
\>[3]{}\Varid{fstLens}{}\<[21]%
\>[21]{}\mathbin{::}\Conid{VLens}\;\Varid{a}\;\Varid{b}\;(\Varid{a},\Varid{y})\;(\Varid{b},\Varid{y}){}\<[E]%
\\
\>[3]{}\Varid{fstLens}\;\Varid{f}\;(\Varid{a},\Varid{y}){}\<[21]%
\>[21]{}\mathrel{=}(\lambda \Varid{b}\to (\Varid{b},\Varid{y}))\mathbin{`\FN{fmap}`}(\Varid{f}\;\Varid{a}){}\<[E]%
\ColumnHook
\end{hscode}\resethooks

Above we are constructing a \ensuremath{\Conid{VLens}} value but the argument applies
equally well to a \ensuremath{\Conid{KLens}} value.  The pair type is functorial
in two arguments.  For \ensuremath{\Varid{fstLens}}, we care about pairs being
functorial with respect to the first position.  If we were required to
pass a $J$ functor explicitly to \ensuremath{\Conid{VLens}}, we would need to add a
wrapper around \ensuremath{(\Varid{a},\Varid{b})} to make it explicitly a functor of the
first position.  Furthermore, we are implicitly using the identity
functor for the $K$ functor.  If we were required to pass a $K$
functor explicitly to \ensuremath{\Conid{VLens}} we would have to wrap and unwrap
the \ensuremath{\Conid{Identity}} functor in Haskell in order to use the lens.
Fortunately, all lens functionality can be implemented without
explicitly mentioning the functor parameters.

The third thing to note about the \ensuremath{\Conid{VLens}} formulation is that we use a
type alias rather than a \ensuremath{\mathbf{newtype}}.  This allows us to compose a lens
of type \ensuremath{\Conid{VLens}\;\Varid{ja}\;\Varid{jb}\;\Varid{ka}\;\Varid{kb}} and another lens of type \ensuremath{\Conid{VLens}\;\Varid{ka}\;\Varid{kb}\;\Varid{la}\;\Varid{lb}} by simply using the standard function composition operator.
There is another advantage that the type alias gives us, which we will see
later.

The isomorphism between the two representations can be written out
explicitly in Haskell.

\begin{hscode}\SaveRestoreHook
\column{B}{@{}>{\hspre}l<{\hspost}@{}}%
\column{3}{@{}>{\hspre}l<{\hspost}@{}}%
\column{5}{@{}>{\hspre}l<{\hspost}@{}}%
\column{20}{@{}>{\hspre}l<{\hspost}@{}}%
\column{E}{@{}>{\hspre}l<{\hspost}@{}}%
\>[3]{}\mathbf{instance}\;\FN{Functor}\;(\Conid{PStore}\;\Varid{i}\;\Varid{j})\;\mathbf{where}{}\<[E]%
\\
\>[3]{}\hsindent{2}{}\<[5]%
\>[5]{}\FN{fmap}\;\Varid{f}\;(\Conid{PStore}\;\Varid{h}\;\Varid{x})\mathrel{=}\Conid{PStore}\;(\Varid{f}\circ\Varid{h})\;\Varid{x}{}\<[E]%
\\[\blanklineskip]%
\>[3]{}\Varid{kLens2VLens}{}\<[20]%
\>[20]{}\mathbin{::}\Conid{KLens}\;\Varid{ja}\;\Varid{jb}\;\Varid{ka}\;\Varid{kb}\to \Conid{VLens}\;\Varid{ja}\;\Varid{jb}\;\Varid{ka}\;\Varid{kb}{}\<[E]%
\\
\>[3]{}\Varid{kLens2VLens}\;\Varid{k}\;\Varid{f}{}\<[20]%
\>[20]{}\mathrel{=}(\lambda (\Conid{PStore}\;\Varid{h}\;\Varid{x})\to \Varid{h}\mathbin{`\FN{fmap}`}\Varid{f}\;\Varid{x})\circ\Varid{k}{}\<[E]%
\\[\blanklineskip]%
\>[3]{}\Varid{vLens2KLens}{}\<[20]%
\>[20]{}\mathbin{::}\Conid{VLens}\;\Varid{ja}\;\Varid{jb}\;\Varid{ka}\;\Varid{kb}\to \Conid{KLens}\;\Varid{ja}\;\Varid{jb}\;\Varid{ka}\;\Varid{kb}{}\<[E]%
\\
\>[3]{}\Varid{vLens2KLens}\;\Varid{v}{}\<[20]%
\>[20]{}\mathrel{=}\Varid{v}\;(\Conid{PStore}\;\Varid{id}){}\<[E]%
\ColumnHook
\end{hscode}\resethooks

The generalised lens representation theorem gives us pairs of
representations of various lens derivatives.  Using pointed functors,
i.e. using the free pointed functor generated by \ensuremath{\Conid{PStore}} in the case of the
coalgebra representation, or quantifying over pointed functors in the
case of the Van Laarhoven representation, gives us the notion of a
partial lens~\cite{oconnorr:2013}, also known as an affine
traversal~\cite{ekmett:2013}.\footnote{An affine traversal from \ensuremath{\Conid{A}} to \ensuremath{\Conid{B}} is
so called because it specifies an isomorphism between \ensuremath{\Conid{A}} and \ensuremath{\Conid{F}\;\Conid{B}} for some
affine container \ensuremath{\Conid{F}}, i.e. for some functor \ensuremath{\Conid{F}} where $\ensuremath{\Conid{F}\;\Conid{X}} \cong \ensuremath{\Conid{C}}_1 \times \ensuremath{\Conid{X}} +
\ensuremath{\Conid{C}}_2$.}

\begin{hscode}\SaveRestoreHook
\column{B}{@{}>{\hspre}l<{\hspost}@{}}%
\column{3}{@{}>{\hspre}l<{\hspost}@{}}%
\column{5}{@{}>{\hspre}l<{\hspost}@{}}%
\column{11}{@{}>{\hspre}l<{\hspost}@{}}%
\column{36}{@{}>{\hspre}c<{\hspost}@{}}%
\column{36E}{@{}l@{}}%
\column{39}{@{}>{\hspre}l<{\hspost}@{}}%
\column{E}{@{}>{\hspre}l<{\hspost}@{}}%
\>[3]{}\mathbf{data}\;{}\<[11]%
\>[11]{}\Conid{FreePointedPStore}\;\Varid{a}\;\Varid{b}\;\Varid{x}{}\<[36]%
\>[36]{}\mathrel{=}{}\<[36E]%
\>[39]{}\Conid{Unit}\;\Varid{x}{}\<[E]%
\\
\>[36]{}\mid {}\<[36E]%
\>[39]{}\Conid{FreePointedPStore}\;(\Varid{b}\to \Varid{x})\;\Varid{a}{}\<[E]%
\\[\blanklineskip]%
\>[3]{}\mbox{\onelinecomment  coalgebra representation of partial lenses}{}\<[E]%
\\
\>[3]{}\mathbf{newtype}\;\Conid{KPartialLens}\;\Varid{ja}\;\Varid{jb}\;\Varid{ka}\;\Varid{kb}\mathrel{=}\Conid{KPartialLens}\;(\Varid{ja}\to \Conid{FreePointedPStore}\;\Varid{ka}\;\Varid{kb}\;\Varid{jb}){}\<[E]%
\\[\blanklineskip]%
\>[3]{}\mathbf{class}\;\FN{Functor}\;\Varid{f}\Rightarrow \FN{Pointed}\;\Varid{f}\;\mathbf{where}{}\<[E]%
\\
\>[3]{}\hsindent{2}{}\<[5]%
\>[5]{}\Varid{point}\mathbin{::}\Varid{a}\to \Varid{f}\;\Varid{a}{}\<[E]%
\\[\blanklineskip]%
\>[3]{}\mbox{\onelinecomment  Van Laarhoven representation of partial lenses}{}\<[E]%
\\
\>[3]{}\mathbf{type}\;\Conid{VPartialLens}\;\Varid{ja}\;\Varid{jb}\;\Varid{ka}\;\Varid{kb}\mathrel{=}\forall \Varid{f}\hsforall .~\FN{Pointed}\;\Varid{f}\Rightarrow (\Varid{ka}\to \Varid{f}\;\Varid{kb})\to \Varid{ja}\to \Varid{f}\;\Varid{jb}{}\<[E]%
\ColumnHook
\end{hscode}\resethooks

A partial lens provides a reference to 0 or 1 occurrences of $K$
within $J$.  If we instead use applicative functors (Section~\ref{sec:applicative}), we get a reference to a
sequence of 0 or more occurrences of $K$ within $J$.  This lens
derivative is called a traversal.

\begin{hscode}\SaveRestoreHook
\column{B}{@{}>{\hspre}l<{\hspost}@{}}%
\column{3}{@{}>{\hspre}l<{\hspost}@{}}%
\column{19}{@{}>{\hspre}c<{\hspost}@{}}%
\column{19E}{@{}l@{}}%
\column{22}{@{}>{\hspre}l<{\hspost}@{}}%
\column{37}{@{}>{\hspre}c<{\hspost}@{}}%
\column{37E}{@{}l@{}}%
\column{E}{@{}>{\hspre}l<{\hspost}@{}}%
\>[3]{}\mathbf{data}\;\Conid{FreeApplicativePStore}\;\Varid{a}\;\Varid{b}\;\Varid{x}{}\<[37]%
\>[37]{}\mathrel{=}{}\<[37E]%
\\
\>[3]{}\hsindent{19}{}\<[22]%
\>[22]{}\Conid{Unit}\;\Varid{x}{}\<[E]%
\\
\>[3]{}\hsindent{16}{}\<[19]%
\>[19]{}\mid {}\<[19E]%
\>[22]{}\Conid{FreeApplicativePStore}\;(\Conid{FreeApplicativePStore}\;\Varid{a}\;\Varid{b}\;(\Varid{b}\to \Varid{x}))\;\Varid{a}{}\<[E]%
\\[\blanklineskip]%
\>[3]{}\mbox{\onelinecomment  coalgebra representation of traversals}{}\<[E]%
\\
\>[3]{}\mathbf{newtype}\;\Conid{KTraversal}\;\Varid{ja}\;\Varid{jb}\;\Varid{ka}\;\Varid{kb}\mathrel{=}\Conid{KTraversal}\;(\Varid{ja}\to \Conid{FreeApplicativePStore}\;\Varid{ka}\;\Varid{kb}\;\Varid{jb}){}\<[E]%
\\[\blanklineskip]%
\>[3]{}\mbox{\onelinecomment  Van Laarhoven representation of traversals}{}\<[E]%
\\
\>[3]{}\mathbf{type}\;\Conid{VTraversal}\;\Varid{ja}\;\Varid{jb}\;\Varid{ka}\;\Varid{kb}\mathrel{=}\forall \Varid{f}\hsforall .~\FN{Applicative}\;\Varid{f}\Rightarrow (\Varid{ka}\to \Varid{f}\;\Varid{kb})\to \Varid{ja}\to \Varid{f}\;\Varid{jb}{}\<[E]%
\ColumnHook
\end{hscode}\resethooks

The Haskell implementation of the isomorphism between
\ensuremath{\Conid{KPartialLens}} and \ensuremath{\Conid{VPartialLens}} and the isomorphism between
\ensuremath{\Conid{KTraversal}} and \ensuremath{\Conid{VTraversal}} is left as an exercise to the interested
reader.

The second advantage of using a type synonym for the Van Laarhoven
representation is that values of type \ensuremath{\Conid{VLens}} are values of type
\ensuremath{\Conid{VPartialLens}} and \ensuremath{\Conid{VTraversal}}, while the values of type \ensuremath{\Conid{KLens}} need
to be explicitly converted to \ensuremath{\Conid{KPartialLens}} and \ensuremath{\Conid{KTraversal}}.  If
Haskell's standard library were modified such that \ensuremath{\FN{Pointed}} was a
super-class of \ensuremath{\FN{Applicative}}, then values of type \ensuremath{\Conid{VPartialLens}} would
be of type \ensuremath{\Conid{VTraversal}} as well.

%


\section{The finiteness of traversals}
\label{sec:traversals}

In this section we show another application of the representation
theorem. We show that traversable functors are exactly the finitary
containers. We first introduce the relevant definitions and then
provide the proof.


%
%
\subsection{Applicative functors}
\label{sec:applicative}

The cartesian product gives the category \ensuremath{\Set} a monoidal structure
$(\ensuremath{\Set,\times,\mathbf{1},\alpha ,\lambda ,\rho })$, where
 $\ensuremath{\alpha }_{X,Y,Z}: \ensuremath{\Conid{X}\;\times\;(\Conid{Y}\;\times\;\Conid{Z})\to (\Conid{X}\;\times\;\Conid{Y})\;\times\;\Conid{Z}}$, 
 $\ensuremath{\lambda }_X : \ensuremath{\mathbf{1}\;\times\;\Conid{X}\to \Conid{X}}$, and
 $\ensuremath{\rho }_X : \ensuremath{\Conid{X}\;\times\;\mathbf{1}\to \Conid{X}}$ 
 are natural isomorphisms expressing associativity of the product,
 left unit and right unit, respectively.

\begin{defn}[Applicative functor]
 An \emph{applicative functor} is a functor \ensuremath{\Conid{F}\colon\Set\to \Set} which is
 strong lax monoidal with respect to this monoidal structure. That is,
 it is equipped with a map and a natural transformation:
\[
\begin{array}{lcl@{\qquad\qquad}l}
\ensuremath{\Varid{u}} &:&  \ensuremath{\mathbf{1}\to \Conid{F}\;\mathbf{1}} 
               &\text{(monoidal unit)}\\
\ensuremath{\star}_{X,Y} &:& \ensuremath{\Conid{F}\;\Conid{X}\;\times\;\Conid{F}\;\Conid{Y}\to \Conid{F}\;(\Conid{X}\;\times\;\Conid{Y})}
               &\text{(monoidal action)}
\end{array}
\]%
such that
\[
\xymatrix{
  \ensuremath{\mathbf{1}\;\times\;\Conid{F}\;\Conid{X}} \ar[d]_{\ensuremath{\Varid{u}\;\times\;\Conid{F}\;\Conid{X}}} 
                \ar[r]^-{\ensuremath{\lambda }}
  &
  \ensuremath{\Conid{F}\;\Conid{X}} 
        \ar@{=}[dd]
 &
  \ensuremath{\Conid{F}\;\Conid{X}\;\times\;\mathbf{1}} \ar[d]^{\ensuremath{\Conid{F}\;\Conid{X}\;\times\;\Varid{u}}}
                 \ar[l]_-{\ensuremath{\rho }}
\\
 \ensuremath{\Conid{F}\;\mathbf{1}\;\times\;\Conid{F}\;\Conid{X}} \ar[d]_{\ensuremath{\star}}
 & &
 \ensuremath{\Conid{F}\;\Conid{X}\;\times\;\Conid{F}\;\mathbf{1}} \ar[d]^{\ensuremath{\star}}
\\
 \ensuremath{\Conid{F}\;(\mathrm{1}\;\times\;\Conid{X})} \ar[r]_-{\ensuremath{\Conid{F}\;\lambda }} 
 &
 \ensuremath{\Conid{F}\;\Conid{X}}
&
 \ensuremath{\Conid{F}\;(\Conid{X}\;\times\;\mathbf{1})} \ar[l]^-{\ensuremath{\Conid{F}\;\rho }}
}
\]%
\[
\xymatrix@C+=2cm{
 \ensuremath{\Conid{F}\;\Conid{X}\;\times\;(\Conid{F}\;\Conid{Y}\;\times\;\Conid{F}\;\Conid{Z})} \ar[d]_{\alpha} 
                             \ar[r]^{\ensuremath{\Conid{F}\;\Conid{X}\;\times\star}}
 &
 \ensuremath{\Conid{F}\;\Conid{X}\;\times\;\Conid{F}\;(\Conid{Y}\;\times\;\Conid{Z})}  \ar[r]^{\ensuremath{\star}}
 &
 \ensuremath{\Conid{F}\;(\Conid{X}\;\times\;(\Conid{Y}\;\times\;\Conid{Z}))}  \ar[d]^{\ensuremath{\Conid{F}\;\alpha }}
 \\
 \ensuremath{(\Conid{F}\;\Conid{X}\;\times\;\Conid{F}\;\Conid{Y})\;\times\;\Conid{F}\;\Conid{Z}} \ar[r]_{\ensuremath{\star\times\;\Conid{F}\;\Conid{Z}}}
 &
 \ensuremath{\Conid{F}\;(\Conid{X}\;\times\;\Conid{Y})\;\times\;\Conid{F}\;\Conid{Z}} \ar[r]_{\ensuremath{\star}}
 &
 \ensuremath{\Conid{F}\;((\Conid{X}\;\times\;\Conid{Y})\;\times\;\Conid{Z})} 
}
\]%

All \ensuremath{\Set} functors are strong, but the strength \ensuremath{\tau \colon\Conid{F}\;\Conid{X}\;\times\;\Conid{Y}\to \Conid{F}\;(\Conid{X}\;\times\;\Conid{Y})} of an applicative functor \ensuremath{\Conid{F}} is required to be
coherent with the monoidal action, i.e. the following diagram
commutes.
\[
\xymatrix@C+=2cm{
     \ensuremath{(\Conid{F}\;\Conid{X}\;\times\;\Conid{F}\;\Conid{Y})\;\times\;\Conid{Z}} \ar[r]^{\alpha^{-1}}\ar[d]_{\ensuremath{\star\times\;\Conid{Z}}} &
         \ensuremath{\Conid{F}\;\Conid{X}\;\times\;(\Conid{F}\;\Conid{Y}\;\times\;\Conid{Z})}\ar[r]^{\ensuremath{\Conid{F}\;\Conid{X}\;\times\;\tau }} & 
         \ensuremath{\Conid{F}\;\Conid{X}\;\times\;\Conid{F}\;(\Conid{Y}\;\times\;\Conid{Z})}\ar[d]^{\ensuremath{\star}}\\
  \ensuremath{\Conid{F}\;(\Conid{X}\;\times\;\Conid{Y})\;\times\;\Conid{Z}} \ar[r]^{\ensuremath{\tau }} &
 \ensuremath{\Conid{F}\;((\Conid{X}\;\times\;\Conid{Y})\;\times\;\Conid{Z})} \ar[r]^{\ensuremath{\Conid{F}\;\alpha }^{-1}} & \ensuremath{\Conid{F}\;(\Conid{X}\;\times\;(\Conid{Y}\;\times\;\Conid{Z}))}}
\]%
\end{defn}

Applicative functors may alternatively be given as a mapping of
objects $\ensuremath{\Conid{F}\colon|\Set|\to |\Set|}$ equipped with two natural transformations
$\ensuremath{\Varid{pure}}_{\ensuremath{\Conid{X}}} : \ensuremath{\Conid{X}\to \Conid{F}\;\Conid{X}}$ and $\ensuremath{\circledast}_{X,Y} : F\,(X \tto Y) \times
F\,X \to F\, Y$, together with some equations
(see~\cite{mcbride08:applicative-programming} for details). This
presentation is more useful for programming and therefore is the one
chosen in Haskell. However, for our purposes, the presentation of
applicative functors as monoidal functors is more convenient. This situation where one presentation is more apt for programming, and another presentation is better for formal reasoning also occurs with monads, where \ensuremath{\Varid{bind}} \ensuremath{(\bind )} is preferred for programming and the multiplication \ensuremath{(\Varid{join})}  is preferred for formal reasoning.  



\begin{defn}[Applicative morphism]
  Let \ensuremath{\Conid{F}} and \ensuremath{\Conid{G}} be applicative functors. An \emph{applicative
    morphism} is a natural transformation \ensuremath{\tau \colon\Conid{F}\to \Conid{G}} that
  respects the unit and multiplication. That is, a natural
  transformation \ensuremath{\tau } such that the following diagrams commute.
\[
\xymatrix{
& \ensuremath{\mathbf{1}} \ar[dl]_{\ensuremath{\Varid{u}}^\ensuremath{\Conid{F}}} \ar[dr]^{\ensuremath{\Varid{u}}^\ensuremath{\Conid{G}}}& \\
\ensuremath{\Conid{F}\;\mathbf{1}} \ar[rr]_{\ensuremath{\tau }_\ensuremath{\mathbf{1}}} & & \ensuremath{\Conid{G}\;\mathbf{1}}
}
\qquad
\xymatrix@C+=2cm{
\ensuremath{\Conid{F}\;\Conid{X}\;\times\;\Conid{F}\;\Conid{Y}}
    \ar[r]^{\ensuremath{\star}^\ensuremath{\Conid{F}}_{\ensuremath{\Conid{X},\Conid{Y}}}}
    \ar[d]_{\ensuremath{\tau }_\ensuremath{\Conid{X}} \ensuremath{\times\;\tau }_\ensuremath{\Conid{Y}}}
 &
\ensuremath{\Conid{F}\;(\Conid{X}\;\times\;\Conid{Y})}
    \ar[d]^{\ensuremath{\tau }_{\ensuremath{\Conid{X}\;\times\;\Conid{Y}}}}
\\
\ensuremath{\Conid{G}\;\Conid{X}\;\times\;\Conid{G}\;\Conid{Y}} 
  \ar[r]_{\ensuremath{\star}^\ensuremath{\Conid{G}}_{\ensuremath{\Conid{X},\Conid{Y}}}}
&
\ensuremath{\Conid{G}\;(\Conid{X}\;\times\;\Conid{Y})}
}
\]%
\end{defn}

Applicative functors and applicative morphisms form a strict monoidal
category $\App$. 
%
%
The identity functor is an applicative functor, and the composition of
applicative functors is an applicative functor. Hence, $\cA$ has the
structure of a strict monoidal category.
%
%
%
%
%
\subsection{Traversable functors}

McBride and Paterson~\shortcite{mcbride08:applicative-programming}
characterise \emph{traversable functors} as those equipped with a
family of morphisms $\ensuremath{\traverse}_{F,A,B} : (A \tto F B) \times T A \to F
(T B)$, natural in an applicative functor $F$, and sets $A$ and $B$ (cf. the type synonym \ensuremath{\Conid{VTraversable}} from Section~\ref{sec:lenses_in_Haskell}.)
However, without further constraints this characterisation is too
coarse. Hence, Jaskelioff and Ryp\'a\v{c}ek~\shortcite{Jaskelioff:MSFP2012}
proposed the following notion:

%



\begin{defn}[Traversable functor]
A functor $T : Set \to Set$  is said to be \emph{traversable}
if there is a family of functions
\[ \traverse_{F,A,B} : (A \tto F B) \times T A \to F (T B) \]%
natural in $F$, $A$,
and $B$ that respects the monoidal structure of applicative functor
composition.
 More concretely, for all applicative functors $F, G : Set \to Set$ and
applicative morphisms $\alpha : F \to G$, the following diagrams
should commute:
\[
\begin{array}{c@{\qquad\qquad}c}
\xymatrix@C+=24mm{ 
   \ensuremath{\Conid{T}\;\Conid{A}} \ar[r]^{\ensuremath{\traverse}_{F,A,B}~\ensuremath{(\Varid{f})}}
    \ar[rd]_{\ensuremath{\traverse}_{G,A,B} (\ensuremath{\alpha }_\ensuremath{\Conid{B}} \ensuremath{\circ\Varid{f}})\quad~~}
   &
   \ensuremath{\Conid{F}\;(\Conid{T}\;\Conid{B})} \ar[d]^{\ensuremath{\alpha }_{\ensuremath{\Conid{T}\;\Conid{B}}}}
   \\ &
  \ensuremath{\Conid{G}\;(\Conid{T}\;\Conid{B})}
}
&  
\xymatrix@C+=10mm{
    & \ensuremath{\Conid{F}\;(\Conid{T}\;(\Conid{G}\;\Conid{B}))} \ar[dr]^{~~~\ensuremath{\Conid{F}\;(\traverse}_{\ensuremath{\Conid{G},\Conid{B},\Conid{C}}}\ensuremath{(\Varid{g})})} \\
   \ensuremath{\Conid{T}\;\Conid{A}} \ar[rr]_{\ensuremath{\traverse}_{\ensuremath{\Conid{FG},\Conid{A},\Conid{C}}}(\ensuremath{\Conid{F}\;\Varid{g}}\circ f)} 
   \ar[ur]^{\ensuremath{\traverse}_{\ensuremath{\Conid{F},\Conid{A},\Conid{GB}}} \ensuremath{(\Varid{f})}~}
    & & \ensuremath{\Conid{F}\;(\Conid{G}\;(\Conid{T}\;\Conid{C}))}
}
\\
  \text{naturality}
& \text{linearity}
\end{array}
\]\[
\begin{array}{c}
\xymatrix{ 
  \ensuremath{\Conid{T}\;(\Conid{Id}\;\Conid{A})} \ar@/_1em/[r]_{\ensuremath{\traverse}_{Id,A,A} (\ensuremath{\Varid{id}}_A)}
           \ar@/^1em/[r]^{\ensuremath{\Varid{id}}_\ensuremath{\Conid{TA}}} &
  \ensuremath{\Conid{Id}\;(\Conid{T}\;\Conid{A})} 
}
\\
\text{unity}
\end{array}
\]%
\end{defn}




%

\subsection{Characterising traversable functors}

Let $\App$ be the category of applicative functors and applicative
morphisms.  In order to prove that traversable functors are finitary
containers, we first note that the forgetful functor $U$ from the
category of applicative functors $\cA$ into the category of endofunctors has
a left adjoint $(-)^*$~\cite{Capriotti2014} and therefore we can apply
Theorem~\ref{theorem:generalized-lens-representation} to any traversal
which satisfies the linearity and unity laws.  Hence for every
traversal on $T$
\[
\traverse_{A,B}: \int_{\F:\App} (A \tto U \F B) \tto T A \tto U\F (T B)
\]%
there is a corresponding coalgebra
\[
    t_{A,B} : T\, A \to U R^*_{A,B}(T\,B)
\]%
where $R^*_{A,B}$ is the free applicative functor for $R_{A,B}$.
%
%
%
The following proposition tells us what this free applicative functor
looks like.

\begin{prop}
 The free applicative functor on $R_{A,B}$ is 
$$R^*_{A,B}\, X = \Sigma\, n : \ensuremath{\mathbb{N}}.\ A^n \times (B^n \to X)$$
with action on morphisms  $R^*_{A,B}(h)\,(n,as,f) =  (n,as,h \circ f)$,  
and applicative structure:
\[
\begin{array}{rcl}
  u & : & R^*_{A,B}\,1 \\
  u & = & (0,\ensuremath{\ast}, \lambda bs. \ensuremath{\ast})\\
\\
\ensuremath{\star}_{X,Y} & :&  R^*_{A,B}\, X\times R^*_{A,B}\, Y \to R^*_{A,B} (X \times Y) \\
(n,as,f) \ensuremath{\star} (m,as',g) & = &  (n+m,\, \ensuremath{\Varid{as}\plus \Varid{as'}},\, \lambda bs. \ensuremath{(\Varid{f}\;(\FN{take}\;\Varid{n}\;\Varid{bs}),\Varid{g}\;(\FN{drop}\;\Varid{n}\;\Varid{bs}))})
\end{array}
\]%
where we write \ensuremath{\Conid{X}^n} for vectors of length $n$, i.e. the  
$n$-fold product $\overbrace{X \times \cdots \times X}^{n~\text{times}}$,
\ensuremath{\plus } for vector append, and \ensuremath{\FN{take}\;\Varid{n}} and \ensuremath{\FN{drop}\;\Varid{n}} for the
functions that given a vector of size $n+m$ return the first $n$
elements and the last $m$ elements respectively.
\end{prop}
The datatype \ensuremath{\Conid{FreeApplicativePStore}} given in
Section~\ref{sec:lenses_in_Haskell} is a Haskell implementation of the
free applicative functor on $R_{A,B}$, namely $R^*_{A,B}$.






Hence $R^*_{A,B}\,X$ consists of
\begin{enumerate}
\item a natural number, which we call the \emph{dimension},
\item a finite vector, which we call the \emph{position},
\item a function from a finite vector, which allows us to \emph{peek} into new positions.
\end{enumerate}
In order to make it easier to talk about the different components we
define projections: let $r = (n,i,g) : R^*_{A,B}\,X$, then $\ensuremath{\FN{dim}\;\Varid{r}\mathrel{=}\Varid{n}}$,
$\ensuremath{\FN{pos}\;\Varid{r}\mathrel{=}\Varid{i}}$, and $\ensuremath{\FN{peek}\;\Varid{r}\mathrel{=}\Varid{g}}$.



Theorem~\ref{theorem:generalized-lens-representation} tells us that
$UR^*$ is a parameterised comonad with the following counit and
comultiplication operations.
\[
\begin{array}{lcl}
\ensuremath{\varepsilon }_{A,X}   & : & UR^*_{A,A}\, X \to X\\
\ensuremath{\varepsilon }_{A,X} (n,as,f) & = & \ensuremath{\Varid{f}\;\Varid{as}}
\\
\\
\ensuremath{\delta }_{A,B,C,X}   & : & UR^*_{A,C}\, X \to UR^*_{A,B}(UR^*_{B,C}\, X)\\
\ensuremath{\delta }_{A,B,C,X} (n,as,f) & = & (n,\, as,\, \lambda bs. (n, bs, f))
\end{array}
\]%
Furthermore, given a traversal of $T$, a coalgebra for 
$UR^*$, $(T,t)$ is given by
$t_{A,B} =\traverse_{A,B}\, \ensuremath{\FN{wrap}}_{A,B}$, where
\[
\begin{array}{lcl}
\ensuremath{\FN{wrap}}_{A,B}   & : & A \to UR^*_{A,B}\, B\\
\ensuremath{\FN{wrap}}_{A,B}\ \ensuremath{\Varid{a}} & = &(1,a,id_b)   
\end{array}
\]%
In the other direction, given a coalgebra for $UR^*$, $(T,t)$, we
obtain a traversal for $T$:
\[
\ensuremath{\traverse}_{A,B}\ \ensuremath{\Varid{f}\;\Varid{x}} = \ensuremath{\mathbf{let}\;(\Varid{n},\Varid{as},\Varid{g})\mathrel{=}\Varid{t}\;\Varid{x}\;\mathbf{in}}\ F(g)\ \ensuremath{(\FN{collect}_n\;\Varid{f}\;\Varid{as})}  
\]%
where $\ensuremath{\FN{collect}_n\;\Varid{f}}\ (x_1,\dots,x_n) =  f(x_1)\star\dots \star f(x_n)$.

\subsection{Finitary containers}

A \emph{finitary container}~\cite{alti:fossacs03} is given by a set of shapes \ensuremath{\Conid{S}}, and an
arity function \ensuremath{\Varid{ar}\colon\Conid{S}\to \mathbb{N}}.
 The \emph{extension} of a finitary container \ensuremath{(\Conid{S},\Varid{ar})} is a functor 
  \ensuremath{\llbracket\Conid{S},\Varid{ar}\rrbracket\colon\Set\to \Set} defined as follows.
\[
 \ensuremath{\llbracket\Conid{S},\Varid{ar}\rrbracket\;\Conid{X}\mathrel{=}\Sigma\;\Varid{s}\colon\Conid{S}.~\Conid{X}^{(\Varid{ar}\;\Varid{s})}}  
\]%
Given an element of an extension of a finitary container \ensuremath{\Varid{c}\mathrel{=}(\Varid{s},\Varid{xs})\colon\Sigma\;\Varid{s}\colon\Conid{S}.~\Conid{X}^{(\Varid{ar}\;\Varid{s})}}, we define projections \ensuremath{\FN{shape}\;\Varid{c}\mathrel{=}\Varid{s}},
and \ensuremath{\FN{contents}\;\Varid{c}\mathrel{=}\Varid{xs}}.

As an example, lists are given by the finitary container $(\Nat,id_{\Nat})$, where the set of shapes indicates the length of the list. Therefore its extension is
\[
\ensuremath{\llbracket\mathbb{N},\Varid{id}\rrbracket\;\Conid{X}\mathrel{=}\Sigma\;\Varid{n}\colon\mathbb{N}.~\Conid{X}^{\Varid{n}}}.
\]%
Vectors of length $n$ are given by the finitary container $(1,\lambda x.n)$. They have only one shape and have a fixed arity.
Streams are containers~\cite{alti:fossacs03} with exactly one shape, but are not finitary.

\begin{lemma}[Finitary containers are traversable]
  The extension of any finitary container $(S,ar)$ is traversable with
  a canonical traversal given by:
\[
\begin{array}{lcl}
\traverse_{F,X,Y}  & :  &(X \tto F\,Y) \times \ensuremath{\llbracket\Conid{S},\Varid{ar}\rrbracket}\,X \ \to\  F\,\ensuremath{\llbracket\Conid{S},\Varid{ar}\rrbracket}\,Y \\ 
\traverse_{F,X,Y}\,(f,(s,xs)) & = & F (\lambda c.\,(s,c)) (\ensuremath{\FN{collect}}_{ar(s)}\ f\ xs)
\end{array}
\]%
\end{lemma}

\subsection{Finitary containers from coalgebras}

For the first part of our proof we already showed that every traversal
is isomorphic to an $UR^*$-coalgebra. For the second part, we show
that if $(T,t)$ is a $UR^*$-coalgebra then $T$ is a finitary
container.

\begin{theorem}\label{thm:coalgebra-finitary}
  Let \ensuremath{\Conid{X}\colon\Set} and let $(T,t)$ be a coalgebra for
  $UR^*$. That is, \ensuremath{\Conid{T}\colon\Set\to \Set} is a functor and $t_{A,B} : \ensuremath{\Conid{T}\;\Conid{A}\to UR^*_{a,b}\;(\Conid{T}\;\Conid{B})}$ is a family natural in $A$ and dinatural in $B$ such
  that certain laws hold (see
  Definition~\ref{defn:coalgebra_parameterised_comonad}).  Then \ensuremath{\Conid{T}\;\Conid{X}}
  is isomorphic to the extension of the finitary container \ensuremath{\llbracket\Conid{T1},\lambda \Varid{s}.~\FN{dim}\;(\Varid{t}\;\Varid{s})\rrbracket\;\Conid{X}}.
\end{theorem}

\begin{proof}
 We define an isomorphism between \ensuremath{\Conid{T}\;\Conid{X}} and \ensuremath{\Sigma\;\Varid{s}\colon\Conid{T1}.~\Conid{X}^{(\FN{dim}\;(\Varid{t}\;\Varid{s}))}}.

  Given a value \ensuremath{\Varid{x}\colon\Conid{T}\;\Conid{X}}, the contents of the resulting container are
  simply the position of \ensuremath{(\Varid{t}\;\Varid{x})}.  The shape of the resulting
  container is obtained by peeking into \ensuremath{(\Varid{t}\;\Varid{x})} at the trivial vector
  \ensuremath{\ast^{\Varid{n}}\colon\mathrm{1}^n} where \ensuremath{\Varid{n}} is the dimension of \ensuremath{(\Varid{t}\;\Varid{x})}.
  More formally we define one direction of the isomorphism as
\begin{hscode}\SaveRestoreHook
\column{B}{@{}>{\hspre}l<{\hspost}@{}}%
\column{9}{@{}>{\hspre}c<{\hspost}@{}}%
\column{9E}{@{}l@{}}%
\column{12}{@{}>{\hspre}l<{\hspost}@{}}%
\column{E}{@{}>{\hspre}l<{\hspost}@{}}%
\>[B]{}\Phi {}\<[9]%
\>[9]{}\colon{}\<[9E]%
\>[12]{}\Conid{T}\;\Conid{X}\to \Sigma\;\Varid{s}\colon\Conid{T1}.~\Conid{X}^{(\FN{dim}\;(\Varid{t}\;\Varid{s}))}{}\<[E]%
\\
\>[B]{}\Phi \;\Varid{x}{}\<[9]%
\>[9]{}\mathrel{=}{}\<[9E]%
\>[12]{}\mathbf{let}\;(\Varid{n},\Varid{i},\Varid{g})\mathrel{=}\Varid{t}\;\Varid{x}\;\mathbf{in}\;(\Varid{g}\;(\ast^{\Varid{n}}),\Varid{i}){}\<[E]%
\ColumnHook
\end{hscode}\resethooks

Given a value \ensuremath{(\Varid{s},\Varid{v})\colon\Sigma\;\Varid{s}\colon\Conid{T1}.~\Conid{X}^{(\FN{dim}\;(\Varid{t}\;\Varid{s}))}} we can
create a \ensuremath{\Conid{T}\;\Conid{X}} by peaking into \ensuremath{(\Varid{t}\;\Varid{s})} at $v$.  More formally, the other
direction of the isomorphism is defined as
\begin{hscode}\SaveRestoreHook
\column{B}{@{}>{\hspre}l<{\hspost}@{}}%
\column{12}{@{}>{\hspre}c<{\hspost}@{}}%
\column{12E}{@{}l@{}}%
\column{15}{@{}>{\hspre}l<{\hspost}@{}}%
\column{E}{@{}>{\hspre}l<{\hspost}@{}}%
\>[B]{}\Psi {}\<[12]%
\>[12]{}\colon{}\<[12E]%
\>[15]{}\Sigma\;\Varid{s}\colon\Conid{T1}.~\Conid{X}^{(\FN{dim}\;(\Varid{t}\;\Varid{s}))}\to \Conid{T}\;\Conid{X}{}\<[E]%
\\
\>[B]{}\Psi \;(\Varid{s},\Varid{v}){}\<[12]%
\>[12]{}\mathrel{=}{}\<[12E]%
\>[15]{}\FN{peek}\;(\Varid{t}\;\Varid{s})\;\Varid{v}{}\<[E]%
\ColumnHook
\end{hscode}\resethooks

First we prove that \ensuremath{\Psi \;(\Phi \;\Varid{x})\mathrel{=}\Varid{x}}.
\begin{calculation}
  \begin{array}{cl}
    & \ensuremath{\Psi \;(\Phi \;\Varid{x})} \\ 
    =  &  \comment{definition of \ensuremath{\Psi }, \ensuremath{\Phi }}\\
    & \ensuremath{\mathbf{let}\;(\Varid{n},\Varid{i},\Varid{g})\mathrel{=}\Varid{t}\;\Varid{x}\;\mathbf{in}\;\FN{peek}\;(\Varid{t}\;(\Varid{g}\;(\ast^{\Varid{n}})))\;\Varid{i}}\\
    =  & \comment{map on morphisms of \ensuremath{UR^*_{a,b}}} \\
    & \ensuremath{\mathbf{let}\;(\Varid{n},\Varid{i},\Varid{h})\mathrel{=}UR^*_{a,b}\;(\Varid{t})\;(\Varid{t}\;\Varid{x})\;\mathbf{in}\;\FN{peek}\;(\Varid{h}\;(\ast^{\Varid{n}}))\;\Varid{i}}\\
    =  &  \comment{comultiplication-coalgebra law}\\
    & \ensuremath{\mathbf{let}\;(\Varid{n},\Varid{i},\Varid{h})\mathrel{=}\delta \;(\Varid{t}\;\Varid{x})\;\mathbf{in}\;\FN{peek}\;(\Varid{h}\;(\ast^{\Varid{n}}))\;\Varid{i}}\\
    =  &  \comment{definition of $\delta$ and \ensuremath{\FN{peek}}}\\
    & \ensuremath{\mathbf{let}\;(\Varid{n},\Varid{i},\Varid{g})\mathrel{=}(\Varid{t}\;\Varid{x})\;\mathbf{in}\;\Varid{g}\;\Varid{i}}\\
    =  &  \comment{definition of $\epsilon$}\\
    & \ensuremath{\varepsilon \;(\Varid{t}\;\Varid{x})}\\
    =  &  \comment{counit-coalgebra law}\\
    & x
  \end{array}
\end{calculation}


Last we prove that \ensuremath{\Phi \;(\Psi \;(\Varid{s},\Varid{v}))\mathrel{=}(\Varid{s},\Varid{v})}.

\begin{calculation}
  \begin{array}[b]{cl}
    & \ensuremath{\Phi \;(\Psi \;(\Varid{s},\Varid{v}))} \\
    =  & \comment{definition of \ensuremath{\Psi }, \ensuremath{\Phi }, and map on morphisms of \ensuremath{UR^*_{a,b}}}\\
    & \ensuremath{\mathbf{let}\;\{\mskip1.5mu (\anonymous ,\anonymous ,\Varid{h})\mathrel{=}UR^*_{a,b}\;\Varid{t}\;(\Varid{t}\;\Varid{s});(\Varid{n},\Varid{i},\Varid{g})\mathrel{=}\Varid{h}\;\Varid{v}\mskip1.5mu\}\;\mathbf{in}\;(\Varid{g}\;(\ast^{\Varid{n}}),\Varid{i})}\\
    =  & \comment{comultiplication-coalgebra law}\\
    & \ensuremath{\mathbf{let}\;\{\mskip1.5mu (\anonymous ,\anonymous ,\Varid{h})\mathrel{=}\delta \;(\Varid{t}\;\Varid{s});(\Varid{n},\Varid{i},\Varid{g})\mathrel{=}\Varid{h}\;\Varid{v}\mskip1.5mu\}\;\mathbf{in}\;(\Varid{g}\;(\ast^{\Varid{n}}),\Varid{i})}\\
    =  & \comment{definition of $\delta$}\\
    & \ensuremath{\mathbf{let}\;(\Varid{n},\Varid{j},\Varid{g})\mathrel{=}\Varid{t}\;\Varid{s}\;\mathbf{in}\;(\Varid{g}\;(\ast^{\Varid{n}}),\Varid{v})}\\
    =  & \comment {\ensuremath{\Varid{j}\mathrel{=}(\ast^{\Varid{n}})} because $1^n$ has a unique element}\\
    & \ensuremath{\mathbf{let}\;(\Varid{n},\Varid{j},\Varid{g})\mathrel{=}\Varid{t}\;\Varid{s}\;\mathbf{in}\;(\Varid{g}\;\Varid{j},\Varid{v})}\\
    =  &\comment{definition of $\epsilon$}\\
    & \ensuremath{(\varepsilon \;(\Varid{t}\;\Varid{s}),\Varid{v})}\\
    =  & \comment{counit-coalgebra law}\\
    & (s,v)
  \end{array}
\end{calculation}
\end{proof}

\bigskip
\begin{corollary}
Let \ensuremath{\Conid{X}\colon\Set} and \ensuremath{\Conid{T}\colon\Set\to \Set} be a traversable functor.
Then \ensuremath{\Conid{T}\;\Conid{X}} is isomorphic to the finitary container \ensuremath{\llbracket\Conid{T1},\lambda \Varid{s}.~\FN{dim}\;(\traverse\;\FN{wrap}\;\Varid{s})\rrbracket\;\Conid{X}}.
\end{corollary}
\begin{proof}
Apply Theorem~\ref{thm:coalgebra-finitary} with the $UR^*$-coalgebra
$t = \traverse\, \ensuremath{\FN{wrap}}$.
\end{proof}

All that remains to show is that this isomorphism maps the traversal
of \ensuremath{\Conid{T}} to the canonical traversal of the finitary container.

\begin{theorem}
  Let \ensuremath{\Conid{T}\colon\Set\to \Set} be a traversable functor and let \ensuremath{\Phi \colon\Conid{T}\;\Conid{X}\to \llbracket\Conid{T1},\lambda \Varid{s}.~\FN{dim}\;(\traverse\;\FN{wrap}\;\Varid{s})\rrbracket\;\Conid{X}} be the isomorphism
  defined above.  Let \ensuremath{\Conid{F}} be an arbitrary applicative functor and let
  \ensuremath{\Varid{f}\colon\Conid{A}\to \Conid{F}\;\Conid{B}} and \ensuremath{\Varid{x}\colon\Conid{T}\;\Conid{A}}.  Then, \ensuremath{\Conid{F}\;(\Phi )\;(\traverse\;\Varid{f}\;\Varid{x})\mathrel{=}\traverse\;\Varid{f}\;(\Phi \;\Varid{x})}.
\end{theorem}
\begin{proof}
Before beginning we prove two small lemmas.
First that \ensuremath{\FN{pos}\;(\traverse\;\FN{wrap}\;\Varid{x})\mathrel{=}\FN{contents}\;(\Phi \;\Varid{x})}.

\begin{calculation}
  \begin{array}{cl}
    & \ensuremath{\FN{pos}\;(\traverse\;\FN{wrap}\;\Varid{x})}\\
    = & \comment {definition of \ensuremath{\FN{pos}}}\\
    & \ensuremath{\mathbf{let}\;(\anonymous ,\Varid{i},\anonymous )\mathrel{=}\traverse\;\FN{wrap}\;\Varid{x}\;\mathbf{in}\;\Varid{i}}\\
    = & \comment{definition of \ensuremath{\Phi }}\\
    & \ensuremath{\FN{contents}\;(\Phi \;\Varid{x})}  
  \end{array}
\end{calculation}
%

\noindent Second, we prove that 
\ensuremath{\Phi \;(\FN{peek}\;(\traverse\;\FN{wrap}\;\Varid{x})\;\Varid{w})\mathrel{=}(\FN{shape}\;(\Phi \;\Varid{x}),\Varid{w})}
\begin{calculation}
  \begin{array}{cl}
    &  \ensuremath{\Phi \;(\FN{peek}\;(\traverse\;\FN{wrap}\;\Varid{x})\;\Varid{w})}\\
    = &  \comment{definition of \ensuremath{\FN{peek}}}\\
    &  \ensuremath{\mathbf{let}\;(\anonymous ,\anonymous ,\Varid{g})\mathrel{=}\traverse\;\FN{wrap}\;\Varid{x}\;\mathbf{in}\;\Phi \;(\Varid{g}\;\Varid{w})}\\
    = &  \comment{definition of \ensuremath{\Phi }}\\
    &  \ensuremath{\mathbf{let}\;\{\mskip1.5mu (\anonymous ,\anonymous ,\Varid{g})\mathrel{=}\traverse\;\FN{wrap}\;\Varid{x};(\Varid{n},\Varid{i},\Varid{h})\mathrel{=}\traverse\;\FN{wrap}\;(\Varid{g}\;\Varid{w})\mskip1.5mu\}\;\mathbf{in}\;(\Varid{h}\;(\ast^{\Varid{n}}),\Varid{i})}\\
    = &  \comment{definition of \ensuremath{UR^*_{a,b}}}\\
    &  \ensuremath{\mathbf{let}\;\{\mskip1.5mu (\anonymous ,\anonymous ,\Varid{g})\mathrel{=}UR^*_{a,b}\;(\traverse\;\FN{wrap})\;(\traverse\;\FN{wrap}\;\Varid{x});(\Varid{n},\Varid{i},\Varid{h})\mathrel{=}\Varid{g}\;\Varid{w}\mskip1.5mu\}\;\mathbf{in}\;(\Varid{h}\;(\ast^{\Varid{n}}),\Varid{i})}\\
    = &  \comment{coalgebra law for $\delta$}\\
    &  \ensuremath{\mathbf{let}\;\{\mskip1.5mu (\anonymous ,\anonymous ,\Varid{g})\mathrel{=}\delta \;(\traverse\;\FN{wrap}\;\Varid{x});(\Varid{n},\Varid{i},\Varid{h})\mathrel{=}\Varid{g}\;\Varid{w}\mskip1.5mu\}\;\mathbf{in}\;(\Varid{h}\;(\ast^{\Varid{n}}),\Varid{i})}\\
    = &  \comment{definition of $\delta$}\\
    &  \ensuremath{\mathbf{let}\;(\anonymous ,\anonymous ,\Varid{g})\mathrel{=}\traverse\;\FN{wrap}\;\Varid{x}\;\mathbf{in}\;(\Varid{g}\;(\ast^{\Varid{n}}),\Varid{w})}\\
    = &  \comment{definition of \ensuremath{\Phi }}\\
    &  \ensuremath{(\FN{shape}\;(\Phi \;\Varid{x}),\Varid{w})}
  \end{array}
\end{calculation}

\noindent Lastly, we prove our main result.

\begin{calculation}
  \begin{array}[b]{cl}
    &  \ensuremath{\Conid{F}\;(\Phi )\;(\traverse\;\Varid{f}\;\Varid{x})}\\
    = &  \comment{isomorphism in Theorem~\ref{theorem:generalized-lens-representation}}\\
    &  \ensuremath{\mathbf{let}\;(\Varid{n},\Varid{i},\Varid{g})\mathrel{=}\traverse\;\FN{wrap}\;\Varid{x}\;\mathbf{in}\;\Conid{F}\;(\Phi )\;(\Conid{F}\;(\Varid{g})\;(\FN{collect}_n\;\Varid{f}\;\Varid{i}))}\\
    = &  \comment{functors respect composition}\\
    &  \ensuremath{\mathbf{let}\;(\Varid{n},\Varid{i},\Varid{g})\mathrel{=}\traverse\;\FN{wrap}\;\Varid{x}\;\mathbf{in}\;\Conid{F}\;(\Phi \circ\Varid{g})\;(\FN{collect}_n\;\Varid{f}\;\Varid{i})}\\
    = &  \comment{application of above two lemmas}\\
    &  \ensuremath{\mathbf{let}\;(\Varid{s},\Varid{v})\mathrel{=}\Phi \;\Varid{x}\;\mathbf{in}\;\Conid{F}\;(\lambda \Varid{c}.~(\Varid{s},\Varid{c}))\;(\FN{collect}_n\;\Varid{f}\;\Varid{v})}\\
    = &  \comment{definition of canonical traverse for finitary containers}\\
    &  \ensuremath{\traverse\;\Varid{f}\;(\Phi \;\Varid{x})}
  \end{array}
\end{calculation}
\end{proof}

The isomorphism between \ensuremath{\Conid{T}} and \ensuremath{\llbracket\Conid{T1},\lambda \Varid{s}.~\FN{dim}\;(\traverse\;\FN{wrap}\;\Varid{s})\rrbracket}
must be natural by construction.  However, naturality is also an
immediate consequence of the preceding theorem because traversing
with the identity functor $\Id$ is equivalent to the mapping on
morphisms of a traversable functor.

%


\section{Implementing algebraic theories}
\label{sec:algebraic_theories}

As a last application of the representation theorem, we take a look at
the case where we consider $\cat{M}$, the category of monads with monad homomorphisms. In
this situation, the functor $(-)^* : \cat{E}\to\cat{M}$, maps any
functor $F:\cat{E}$ to $F^*$, the free monad on $F$, while the functor
$U: \cat{M} \to \cat{E}$ forgets the monad structure. The
representation theorem then states that
\begin{equation}\label{iso-monad}
 \int_{M \in \cat{M}}  (A \tto U\!M\,B) \tto U\!M\,X \quad\cong\quad UR_{A,B}^*\,X   
\end{equation}
where, $R_{A,B}\,X = A \times (B \tto X)$ is the parameterised store
comonad.

In Haskell, we can write the isomorphism~(\ref{iso-monad}) as
\[
\ensuremath{\forall \Varid{m}\hsforall .~\FN{Monad}\;\Varid{m}\Rightarrow (\Varid{a}\to \Varid{m}\;\Varid{b})\to \Varid{m}\;\Varid{x}}
  \quad\cong\quad
\ensuremath{\Conid{Free}\;(\Conid{PStore}\;\Varid{a}\;\Varid{b})\;\Varid{x}} 
\]%
where \ensuremath{\Conid{PStore}} (as given in Section~\ref{sec:lenses_in_Haskell}) and the free monad construction are as follows:
\begin{hscode}\SaveRestoreHook
\column{B}{@{}>{\hspre}l<{\hspost}@{}}%
\column{3}{@{}>{\hspre}l<{\hspost}@{}}%
\column{20}{@{}>{\hspre}l<{\hspost}@{}}%
\column{E}{@{}>{\hspre}l<{\hspost}@{}}%
\>[B]{}\mathbf{newtype}\;\Conid{PStore}\;\Varid{a}\;\Varid{b}\;\Varid{x}\mathrel{=}\Conid{PStore}\;(\Varid{b}\to \Varid{x})\;\Varid{a}{}\<[E]%
\\[\blanklineskip]%
\>[B]{}\mathbf{data}\;\Conid{Free}\;\Varid{f}\;\Varid{x}\mathrel{=}\Conid{Unit}\;\Varid{x}\mid \Conid{Branch}\;(\Varid{f}\;(\Conid{Free}\;\Varid{f}\;\Varid{x})){}\<[E]%
\\[\blanklineskip]%
\>[B]{}\mathbf{instance}\;\FN{Functor}\;\Varid{f}\Rightarrow \FN{Monad}\;(\Conid{Free}\;\Varid{f})\;\mathbf{where}{}\<[E]%
\\
\>[B]{}\hsindent{3}{}\<[3]%
\>[3]{}\Varid{return}{}\<[20]%
\>[20]{}\mathrel{=}\Conid{Unit}{}\<[E]%
\\
\>[B]{}\hsindent{3}{}\<[3]%
\>[3]{}\Conid{Pure}\;\Varid{x}\bind \Varid{f}{}\<[20]%
\>[20]{}\mathrel{=}\Varid{f}\;\Varid{x}{}\<[E]%
\\
\>[B]{}\hsindent{3}{}\<[3]%
\>[3]{}\Conid{Branch}\;\Varid{xs}\bind \Varid{f}{}\<[20]%
\>[20]{}\mathrel{=}\Conid{Branch}\;(\FN{fmap}\;(\bind \Varid{f})\;\Varid{xs}){}\<[E]%
\ColumnHook
\end{hscode}\resethooks

This way of constructing a free monad from an arbitrary functor
requires a recursive datatype.  The
isomorphism~(\ref{iso-monad}), on the other hand, shows a
non-recursive way of describing the free monad on functors of the form
\ensuremath{\Conid{PStore}\;\Varid{a}\;\Varid{b}}.

While this result seems to be of limited applicability, we note that
every signature of an algebraic operation with parameter \ensuremath{\Varid{a}} and arity
\ensuremath{\Varid{b}} determines a functor of this form. Hence, the theorem tells us how
to construct the free monad on a given signature of a single algebraic
operation.  Intuitively the type
\[
\ensuremath{\forall \Varid{m}\hsforall .~\FN{Monad}\;\Varid{m}\Rightarrow (\Varid{a}\to \Varid{m}\;\Varid{b})\to \Varid{m}\;\Varid{x}}
\]%
describes a monadic computation \ensuremath{\Varid{m}\;\Varid{x}} in which the only source of
impurity is the operation of type \ensuremath{\Varid{a}\to \Varid{m}\;\Varid{b}} in the argument. This
type can be implemented in Haskell in the following manner, where we
have abstracted over the types of the argument operation.
\begin{hscode}\SaveRestoreHook
\column{B}{@{}>{\hspre}l<{\hspost}@{}}%
\column{3}{@{}>{\hspre}l<{\hspost}@{}}%
\column{13}{@{}>{\hspre}l<{\hspost}@{}}%
\column{E}{@{}>{\hspre}l<{\hspost}@{}}%
\>[B]{}\mathbf{newtype}\;\Conid{FreeOp}\;\Varid{primOp}\;\Varid{x}\mathrel{=}\Conid{FreeOp}\;\{\mskip1.5mu \Varid{runOp}\mathbin{::}\forall \Varid{m}\hsforall .~\FN{Monad}\;\Varid{m}\Rightarrow \Varid{primOp}\;\Varid{m}\to \Varid{m}\;\Varid{x}\mskip1.5mu\}{}\<[E]%
\\[\blanklineskip]%
\>[B]{}\mathbf{instance}\;\FN{Monad}\;(\Conid{FreeOp}\;\Varid{primOp})\;\mathbf{where}{}\<[E]%
\\
\>[B]{}\hsindent{3}{}\<[3]%
\>[3]{}\Varid{return}\;\Varid{x}{}\<[13]%
\>[13]{}\mathrel{=}\Conid{FreeOp}\;(\Varid{const}\;(\Varid{return}\;\Varid{x})){}\<[E]%
\\
\>[B]{}\hsindent{3}{}\<[3]%
\>[3]{}\Varid{x}\bind \Varid{f}{}\<[13]%
\>[13]{}\mathrel{=}\Conid{FreeOp}\;(\lambda \Varid{op}\to \Varid{runOp}\;\Varid{x}\;\Varid{op}\bind \lambda \Varid{a}\to \Varid{runOp}\;(\Varid{f}\;\Varid{a})\;\Varid{op}){}\<[E]%
\ColumnHook
\end{hscode}\resethooks

Notice that the bind operation for \ensuremath{\Conid{FreeOp}} is not recursive, but
is implemented in terms of the bind operation for an arbitrary
abstract monad.

For example, exceptions in a type \ensuremath{\Varid{e}} can be given by a nullary operation \ensuremath{\Varid{throw}}
with parameter \ensuremath{\Varid{e}}.
\footnote{%
  In order to avoid clutter, we sometimes use a type synonym where a
  real implementation would require a newtype, with its associated
  constructor and destructor.}
\begin{hscode}\SaveRestoreHook
\column{B}{@{}>{\hspre}l<{\hspost}@{}}%
\column{E}{@{}>{\hspre}l<{\hspost}@{}}%
\>[B]{}\mathbf{type}\;\Conid{Exc}\;\Varid{e}\;\Varid{m}\mathrel{=}\Varid{e}\to \Varid{m}\;\emptyset{}\<[E]%
\ColumnHook
\end{hscode}\resethooks
where \ensuremath{\emptyset} is the empty type, and hence \ensuremath{\Conid{FreeOp}\;(\Conid{Exc}\;\Varid{e})} is the type
of monadic computations which can throw an exception using the
following operation:
\begin{hscode}\SaveRestoreHook
\column{B}{@{}>{\hspre}l<{\hspost}@{}}%
\column{10}{@{}>{\hspre}l<{\hspost}@{}}%
\column{E}{@{}>{\hspre}l<{\hspost}@{}}%
\>[B]{}\Varid{throw}{}\<[10]%
\>[10]{}\mathbin{::}\Varid{e}\to \Conid{FreeOp}\;(\Conid{Exc}\;\Varid{e})\;\emptyset{}\<[E]%
\\
\>[B]{}\Varid{throw}\;\Varid{e}{}\<[10]%
\>[10]{}\mathrel{=}\Conid{FreeOp}\;(\lambda \Varid{\char95 throw}\to \Varid{\char95 throw}\;\Varid{e}){}\<[E]%
\ColumnHook
\end{hscode}\resethooks

We may model environments in \ensuremath{\Varid{r}} by an operation \ensuremath{\Varid{ask}} with parameter
\ensuremath{()} and arity \ensuremath{\Varid{r}}.
\begin{hscode}\SaveRestoreHook
\column{B}{@{}>{\hspre}l<{\hspost}@{}}%
\column{E}{@{}>{\hspre}l<{\hspost}@{}}%
\>[B]{}\mathbf{type}\;\Conid{Env}\;\Varid{r}\;\Varid{m}\mathrel{=}()\to \Varid{m}\;\Varid{r}{}\<[E]%
\ColumnHook
\end{hscode}\resethooks
Hence, \ensuremath{\Conid{FreeOp}\;(\Conid{Env}\;\Varid{r})} is the type of monadic computation which can
read an environment using the following operation:
\begin{hscode}\SaveRestoreHook
\column{B}{@{}>{\hspre}l<{\hspost}@{}}%
\column{8}{@{}>{\hspre}l<{\hspost}@{}}%
\column{E}{@{}>{\hspre}l<{\hspost}@{}}%
\>[B]{}\Varid{ask}{}\<[8]%
\>[8]{}\mathbin{::}\Conid{FreeOp}\;(\Conid{Env}\;\Varid{r})\;\Varid{r}{}\<[E]%
\\
\>[B]{}\Varid{ask}{}\<[8]%
\>[8]{}\mathrel{=}\Conid{FreeOp}\;(\lambda \Varid{\char95 ask}\to \Varid{\char95 ask}\;()){}\<[E]%
\ColumnHook
\end{hscode}\resethooks

More generally, we may want to consider algebraic theories with more
than one operation.  Following the same argument as before, but
considering the N-ary representation theorem, we can construct the
free monad on any signature of algebraic operations and express it by
its \emph{generic effects}~\cite{algOpers2003} by means of a
polymorphic type.



For example, a simple teletype interface can be represented by
the following functor~\cite{swierstra2008:la-carte}:

\begin{hscode}\SaveRestoreHook
\column{B}{@{}>{\hspre}l<{\hspost}@{}}%
\column{18}{@{}>{\hspre}c<{\hspost}@{}}%
\column{18E}{@{}l@{}}%
\column{21}{@{}>{\hspre}l<{\hspost}@{}}%
\column{E}{@{}>{\hspre}l<{\hspost}@{}}%
\>[B]{}\mathbf{data}\;\Conid{Teletype}\;\Varid{x}{}\<[18]%
\>[18]{}\mathrel{=}{}\<[18E]%
\>[21]{}\Conid{GetChar}\;(\Conid{Char}\to \Varid{x}){}\<[E]%
\\
\>[18]{}\mid {}\<[18E]%
\>[21]{}\Conid{PutChar}\;\Conid{Char}\;\Varid{x}{}\<[E]%
\ColumnHook
\end{hscode}\resethooks

The free monad generated by this \ensuremath{\Conid{Teletype}} functor produces a tree
representing all the interactions with a teletype machine a user can
have.
%
The \ensuremath{\Conid{Teletype}} functor
is isomorphic to a sum of instances of $R$
\[
  \ensuremath{\Conid{Teletype}\;\Varid{x}} 
     \quad\cong\quad 
  \ensuremath{((),\Conid{Char}\to \Varid{x})\mathbin{+}(\Conid{Char},()\to \Varid{x})}
     \quad\cong\quad 
  \ensuremath{(\Conid{R}\;()\;\Conid{Char}\mathbin{+}\Conid{R}\;\Conid{Char}\;())\;\Varid{x}}
\]%
By the N-ary representation theorem, the free monad generated by \ensuremath{\Conid{Teletype}}
is isomorphic to
\[
\ensuremath{\forall \Varid{m}\hsforall .~\FN{Monad}\;\Varid{m}\Rightarrow (()\to \Varid{m}\;\Conid{Char})\to (\Conid{Char}\to \Varid{m}\;())\to \Varid{m}\;\Varid{x}}
\]%

We define a type for representing teletype operations. In order to
reuse our previous definition of \ensuremath{\Conid{FreeOp}} and to get names for each
argument, we define the type as a record in which each field
corresponds to an operation.

\begin{hscode}\SaveRestoreHook
\column{B}{@{}>{\hspre}l<{\hspost}@{}}%
\column{21}{@{}>{\hspre}c<{\hspost}@{}}%
\column{21E}{@{}l@{}}%
\column{24}{@{}>{\hspre}l<{\hspost}@{}}%
\column{36}{@{}>{\hspre}l<{\hspost}@{}}%
\column{E}{@{}>{\hspre}l<{\hspost}@{}}%
\>[B]{}\mathbf{data}\;\Conid{TTOp}\;\Varid{m}\mathrel{=}\Conid{TTOp}\;{}\<[21]%
\>[21]{}\{\mskip1.5mu {}\<[21E]%
\>[24]{}\Varid{\char95 ttGetChar}{}\<[36]%
\>[36]{}\mathbin{::}\Varid{m}\;\Conid{Char}{}\<[E]%
\\
\>[21]{},{}\<[21E]%
\>[24]{}\Varid{\char95 ttPutChar}{}\<[36]%
\>[36]{}\mathbin{::}\Conid{Char}\to \Varid{m}\;(){}\<[E]%
\\
\>[21]{}\mskip1.5mu\}{}\<[21E]%
\ColumnHook
\end{hscode}\resethooks

We obtain the free monad for \ensuremath{\Conid{TTOp}} and define
operations on it that basically choose the corresponding field from
the record.

\begin{hscode}\SaveRestoreHook
\column{B}{@{}>{\hspre}l<{\hspost}@{}}%
\column{12}{@{}>{\hspre}c<{\hspost}@{}}%
\column{12E}{@{}l@{}}%
\column{14}{@{}>{\hspre}c<{\hspost}@{}}%
\column{14E}{@{}l@{}}%
\column{16}{@{}>{\hspre}l<{\hspost}@{}}%
\column{18}{@{}>{\hspre}l<{\hspost}@{}}%
\column{E}{@{}>{\hspre}l<{\hspost}@{}}%
\>[B]{}\mathbf{type}\;\Conid{FreeTT}\mathrel{=}\Conid{FreeOp}\;\Conid{TTOp}{}\<[E]%
\\[\blanklineskip]%
\>[B]{}\Varid{ttGetChar}{}\<[12]%
\>[12]{}\mathbin{::}{}\<[12E]%
\>[16]{}\Conid{FreeTT}\;\Conid{Char}{}\<[E]%
\\
\>[B]{}\Varid{ttGetChar}{}\<[12]%
\>[12]{}\mathrel{=}{}\<[12E]%
\>[16]{}\Conid{FreeOp}\;\Varid{\char95 ttGetChar}{}\<[E]%
\\[\blanklineskip]%
\>[B]{}\Varid{ttPutChar}{}\<[14]%
\>[14]{}\mathbin{::}{}\<[14E]%
\>[18]{}\Conid{Char}\to \Conid{FreeTT}\;(){}\<[E]%
\\
\>[B]{}\Varid{ttPutChar}\;\Varid{c}{}\<[14]%
\>[14]{}\mathrel{=}{}\<[14E]%
\>[18]{}\Conid{FreeOp}\;(\lambda \Varid{po}\to \Varid{\char95 ttPutChar}\;\Varid{po}\;\Varid{c}){}\<[E]%
\ColumnHook
\end{hscode}\resethooks

Values of type \ensuremath{\Conid{FreeTT}} can easily be interpreted in \ensuremath{\Conid{IO}}, by
providing operations of the appropriate type.

\begin{hscode}\SaveRestoreHook
\column{B}{@{}>{\hspre}l<{\hspost}@{}}%
\column{3}{@{}>{\hspre}l<{\hspost}@{}}%
\column{11}{@{}>{\hspre}l<{\hspost}@{}}%
\column{19}{@{}>{\hspre}c<{\hspost}@{}}%
\column{19E}{@{}l@{}}%
\column{23}{@{}>{\hspre}l<{\hspost}@{}}%
\column{30}{@{}>{\hspre}c<{\hspost}@{}}%
\column{30E}{@{}l@{}}%
\column{33}{@{}>{\hspre}l<{\hspost}@{}}%
\column{45}{@{}>{\hspre}l<{\hspost}@{}}%
\column{E}{@{}>{\hspre}l<{\hspost}@{}}%
\>[B]{}\Varid{runTTIO}\mathbin{::}\Conid{FreeTT}\;\Varid{a}\to \Conid{IO}\;\Varid{a}{}\<[E]%
\\
\>[B]{}\Varid{runTTIO}\mathrel{=}\Varid{runOp}\;\Varid{ttOpIO}{}\<[E]%
\\
\>[B]{}\hsindent{3}{}\<[3]%
\>[3]{}\mathbf{where}\;{}\<[11]%
\>[11]{}\Varid{ttOpIO}{}\<[19]%
\>[19]{}\mathbin{::}{}\<[19E]%
\>[23]{}\Conid{TTOp}\;\Conid{IO}{}\<[E]%
\\
\>[11]{}\Varid{ttOpIO}{}\<[19]%
\>[19]{}\mathrel{=}{}\<[19E]%
\>[23]{}\Conid{TTOp}\;{}\<[30]%
\>[30]{}\{\mskip1.5mu {}\<[30E]%
\>[33]{}\Varid{\char95 ttGetChar}{}\<[45]%
\>[45]{}\mathrel{=}\Varid{getChar}{}\<[E]%
\\
\>[30]{},{}\<[30E]%
\>[33]{}\Varid{\char95 ttPutChar}{}\<[45]%
\>[45]{}\mathrel{=}\Varid{putChar}{}\<[E]%
\\
\>[30]{}\mskip1.5mu\}{}\<[30E]%
\ColumnHook
\end{hscode}\resethooks

Of course, the larger purpose is that \ensuremath{\Conid{FreeTT}} values can be
interpreted in other ways, for example, by logging input, or for use in
automated tests by replaying previously logged input.  Furthermore, a
\ensuremath{\Conid{FreeOp}} monad can easily be embedded into another \ensuremath{\Conid{FreeOp}} monad with
a larger set of primitive commands, or interpreted into another
\ensuremath{\Conid{FreeOp}} monad with a smaller, more primitive set of commands,
providing a simple way of implementing handlers of algebraic
effects~\cite{Pretnar2009}. Hence, Theorem~\ref{thm:representation2}
might provide the basis for a simple implementation of an
algebraic-effects library.

%


\section{Related work}
\label{sec:related}

Traversable functors were introduced by McBride and
Paterson~\shortcite{mcbride08:applicative-programming}, generalising a
notion of traversal by Moggi et al.~\shortcite{MoggiBJ99}. The notion
proposed was too coarse and Gibbons and
Oliveira~\shortcite{Gibbons.Oliveira.Iterator} analysed several
properties that should hold for all traversals. Based on some of these
properties, Jaskelioff and
Ryp\'a\v{c}ek~\shortcite{Jaskelioff:MSFP2012} proposed a
characterisation of traversable functors, and conjectured that they
were isomorphic to finitary containers~\cite{alti:fossacs03}. The
conjecture was proven correct by Bird et al.~\shortcite{BirdGMVS13} by
a means of a change of representation. The proof of this same fact
presented in Section~\ref{sec:traversals} uses a similar change of
representation and was found independently.

The representation of the free applicative functor on the
parameterised store comonad, $R$, is a dependently typed version of
Van Laarhoven's \ensuremath{\Conid{FunList}} data type~\cite{vanLaarhoven:2009a}.  Van
Laarhoven's applicative and parameterised comonad instances for this
type have been translated to work on the dependently typed
implementation.  A particular case of the representation theorem has
been conjectured by Van Laarhoven~\shortcite{vanLaarhoven:2009b}, and
proved by O'Connor~\shortcite{oconnor:2011}.  The proof of
representation theorem for functors via the Yoneda lemma was discovered
independently by Bartosz Milewski~\shortcite{Milewski:2013}.

The representation theorems applied to the case where the structured
functors are monads (as in Section~\ref{sec:algebraic_theories})
yields isomorphisms analogous to the ones presented by Bauer et
al.~\shortcite{Bauer2013}. However, our proof is based on a categorical
model, while theirs is based on a parametric model. Also, as opposed
to us, they do not explore the connection with algebraic effects.

Bernardy et al.~\cite{testing_polymorphic} use a representation
theorem to transform polymorphic properties of a certain shape into
monomorphic properties, which are easier and more efficient to
test. This suggests that another application for the representation
theorems in this article is to facilitate the testing of polymorphic
properties.

%


\section*{Acknowledgements}
Jaskelioff is funded by ANPCyT PICT 2009-15.
Many thanks go to Edward Kmett who assisted the authors with the
isomorphism between \ensuremath{\Conid{KLens}} and \ensuremath{\Conid{VLens}}, and to Exequiel Rivas, Jeremy
Gibbons, and the anonymous referees for helping us improve the
presentation of the paper.  We also thank Shachaf Ben-Kiki for
explaining why affine traversals are called so, and Gabor Greif for
finding some typos. 
%

%

\bibliographystyle{jfp}
\bibliography{traversable}

\providecommand{\noopsort}[1]{}
\begin{thebibliography}{}

\bibitem[\protect\citename{Abbott {\em et~al.}\relax, }2003]{alti:fossacs03}
Abbott, Michael, Altenkirch, Thorsten, \& Ghani, Neil. (2003).
\newblock Categories of containers.
\newblock {\em Pages  23--38 of:} {\em Proceedings of {F}oundations of
  {S}oftware {S}cience and {C}omputation {S}tructures}.

\bibitem[\protect\citename{Atkey, }2009a]{algebras-param-monads}
Atkey, Robert. (2009a).
\newblock Algebras for parameterised monads.
\newblock {\em Pages  3--17 of:} Kurz, Alexander, Lenisa, Marina, \& Tarlecki,
  Andrzej (eds), {\em {A}lgebra and {C}oalgebra in {C}omputer {S}cience,
  {T}hird {I}nternational {C}onference, {CALCO} 2009, {U}dine, {I}taly,
  {S}eptember 7-10, 2009. {P}roceedings}.
\newblock Lecture Notes in Computer Science, vol. 5728.
\newblock Springer.

\bibitem[\protect\citename{Atkey, }2009b]{paramnotions-jfp}
Atkey, Robert. (2009b).
\newblock Parameterised notions of computation.
\newblock {\em Journal of functional programming}, {\bf 19}(3 \& 4), 335--376.

\bibitem[\protect\citename{Awodey, }2006]{awodey2006ct}
Awodey, Steve. (2006).
\newblock {\em {Category theory}}.
\newblock Oxford University Press, USA.

\bibitem[\protect\citename{Bainbridge {\em et~al.}\relax,
  }1990]{BainbridgeFSS90}
Bainbridge, Edwin~S., Freyd, Peter~J., Scedrov, Andre, \& Scott, Philip~J.
  (1990).
\newblock Functorial polymorphism.
\newblock {\em Theoretical computer science}, {\bf 70}(1), 35--64.

\bibitem[\protect\citename{Bauer {\em et~al.}\relax, }2013]{Bauer2013}
Bauer, Andrej, Hofmann, Martin, \& Karbyshev, Aleksandr. (2013).
\newblock On monadic parametricity of second-order functionals.
\newblock {\em Pages  225--240 of:} Pfenning, Frank (ed), {\em Foundations of
  {S}oftware {S}cience and {C}omputation {S}tructures}.
\newblock Lecture Notes in Computer Science, vol. 7794.
\newblock Springer Berlin Heidelberg.

\bibitem[\protect\citename{Bernardy {\em et~al.}\relax,
  }2010]{testing_polymorphic}
Bernardy, Jean-Philippe, Jansson, Patrik, \& Claessen, Koen. (2010).
\newblock Testing polymorphic properties.
\newblock {\em Pages  125--144 of:} Gordon, Andrew~D. (ed), {\em Programming
  languages and systems}.
\newblock Lecture Notes in Computer Science, vol. 6012.
\newblock Springer Berlin Heidelberg.

\bibitem[\protect\citename{Bird \& de~Moor, }1997]{BdM97}
Bird, Richard, \& de~Moor, Oege. (1997).
\newblock {\em Algebra of programming}.
\newblock Upper Saddle River, NJ, USA: Prentice-Hall, Inc.

\bibitem[\protect\citename{Bird {\em et~al.}\relax, }2013]{BirdGMVS13}
Bird, Richard, Gibbons, Jeremy, Mehner, Stefan, Voigtl\"{a}nder, Janis, \&
  Schrijvers, Tom. (2013).
\newblock Understanding idiomatic traversals backwards and forwards.
\newblock {\em Pages  25--36 of:} {\em Proceedings of the 2013 {ACM SIGPLAN}
  {S}ymposium on {H}askell}.
\newblock Haskell '13.
\newblock New York, NY, USA: ACM.

\bibitem[\protect\citename{Capriotti \& Kaposi, }2014]{Capriotti2014}
Capriotti, Paolo, \& Kaposi, Ambros. 2014 (April).
\newblock Free applicative functors.
\newblock  {\em Proceedings of the fifth {W}orkshop on {M}athematically
  {S}tructured {F}unctional {P}rogramming}.
\newblock MSFP '14.

\bibitem[\protect\citename{Danielsson {\em et~al.}\relax,
  }2006]{Danielsson:2006}
Danielsson, Nils~Anders, Hughes, John, Jansson, Patrik, \& Gibbons, Jeremy.
  (2006).
\newblock Fast and loose reasoning is morally correct.
\newblock {\em Sigplan not.}, {\bf 41}(1), 206--217.

\bibitem[\protect\citename{Foster {\em et~al.}\relax, }2007]{foster:2007}
Foster, J.~Nathan, Greenwald, Michael~B., Moore, Jonathan~T., Pierce,
  Benjamin~C., \& Schmitt, Alan. (2007).
\newblock Combinators for bidirectional tree transformations: A linguistic
  approach to the view-update problem.
\newblock {\em Acm trans. program. lang. syst.}, {\bf 29}(3).

\bibitem[\protect\citename{Gibbons \& Johnson, }2012]{colens}
Gibbons, Jeremy, \& Johnson, Mike. (2012).
\newblock Relating algebraic and coalgebraic descriptions of lenses.
\newblock  vol. 49 (Bidirectional Transformations 2012).

\bibitem[\protect\citename{Gibbons \& Oliveira,
  }2009]{Gibbons.Oliveira.Iterator}
Gibbons, Jeremy, \& Oliveira, Bruno c. d.~s. (2009).
\newblock The essence of the iterator pattern.
\newblock {\em Journal of {F}unctional {P}rogramming}, {\bf 19}(July),
  377--402.

\bibitem[\protect\citename{Jaskelioff \& Rypacek, }2012]{Jaskelioff:MSFP2012}
Jaskelioff, Mauro, \& Rypacek, Ondrej. (2012).
\newblock An investigation of the laws of traversals.
\newblock {\em Pages  40--49 of:} Chapman, James, \& Levy, Paul~Blain (eds),
  {\em Proceedings of the {F}ourth {W}orkshop on {M}athematically {S}tructured
  {F}unctional {P}rogramming}.
\newblock EPTCS, vol. 76.

\bibitem[\protect\citename{Kmett, }2013]{ekmett:2013}
Kmett, Edward. 2013 (Oct.).
\newblock {\em lens-4.0: Lenses, folds and traversals}.
\newblock \url{http://ekmett.github.io/lens/Control-Lens-Traversal.html}.

\bibitem[\protect\citename{{\noopsort{Laarhoven}}Van~Laarhoven,
  }2009a]{vanLaarhoven:2009c}
{\noopsort{Laarhoven}}Van~Laarhoven, Twan. 2009a (Aug.).
\newblock {\em {CPS} based functional references}.
\newblock \url{http://twanvl.nl/blog/haskell/cps-functional-references}.

\bibitem[\protect\citename{{\noopsort{Laarhoven}}Van~Laarhoven,
  }2009b]{vanLaarhoven:2009a}
{\noopsort{Laarhoven}}Van~Laarhoven, Twan. 2009b (Apr.).
\newblock {\em A non-regular data type challenge}.
\newblock \url{http://twanvl.nl/blog/haskell/non-regular1}.

\bibitem[\protect\citename{{\noopsort{Laarhoven}}Van~Laarhoven,
  }2009c]{vanLaarhoven:2009b}
{\noopsort{Laarhoven}}Van~Laarhoven, Twan. 2009c (Apr.).
\newblock {\em Where do {I} get my non-regular types?}
\newblock \url{http://twanvl.nl/blog/haskell/non-regular2}.

\bibitem[\protect\citename{Mac~Lane, }1971]{macLaneS:catwm}
Mac~Lane, Saunders. (1971).
\newblock {\em Categories for the working mathematician}.
\newblock Graduate Texts in Mathematics, no. ~5.
\newblock Springer-Verlag.
\newblock Second edition, 1998.

\bibitem[\protect\citename{McBride \& Paterson,
  }2008]{mcbride08:applicative-programming}
McBride, Conor, \& Paterson, Ross. (2008).
\newblock Applicative programming with effects.
\newblock {\em Journal of functional programming}, {\bf 18}(01), 1--13.

\bibitem[\protect\citename{Milewski, }2013]{Milewski:2013}
Milewski, Bartosz. 2013 (Oct.).
\newblock {\em Lenses, stores, and yoneda}.
\newblock \url{http://bartoszmilewski.com/2013/10/08/lenses-stores-and-yoneda}.

\bibitem[\protect\citename{Moggi {\em et~al.}\relax, }1999]{MoggiBJ99}
Moggi, Eugenio, Bell{\`e}, Giana, \& Jay, C.~Barry. (1999).
\newblock Monads, shapely functors and traversals.
\newblock {\em Electronic notes in theoretical computer science}, {\bf 29}, 187
  -- 208.
\newblock CTCS '99, Conference on Category Theory and Computer Science.

\bibitem[\protect\citename{O'Connor, }2010]{oconnor:2010}
O'Connor, Russell. 2010 (Nov.).
\newblock {\em Lenses are exactly the coalgebras for the store comonad}.
\newblock \url{http://r6research.livejournal.com/23705.html}.

\bibitem[\protect\citename{O'Connor, }2011]{oconnor:2011}
O'Connor, Russell. (2011).
\newblock Functor is to lens as applicative is to biplate: Introducing
  multiplate.
\newblock {\em Corr}, {\bf abs/1103.2841v1}.

\bibitem[\protect\citename{O'Connor {\em et~al.}\relax, }2013]{oconnorr:2013}
O'Connor, Russell, Kmett, Edward~A., \& Morris, Tony. 2013 (Oct.).
\newblock {\em data-lens-2.10.4: Haskell 98 lenses}.
\newblock
  \url{http://hackage.haskell.org/package/data-lens-2.10.4/docs/Data-Lens-Partial-Common.html}.

\bibitem[\protect\citename{Plotkin \& Power, }2003]{algOpers2003}
Plotkin, Gordon, \& Power, John. (2003).
\newblock Algebraic operations and generic effects.
\newblock {\em Applied categorical structures}, {\bf 11}(1), 69--94.

\bibitem[\protect\citename{Plotkin \& Pretnar, }2009]{Pretnar2009}
Plotkin, Gordon, \& Pretnar, Matija. (2009).
\newblock Handlers of algebraic effects.
\newblock {\em Pages  80--94 of:} Castagna, Giuseppe (ed), {\em Programming
  languages and systems}.
\newblock Lecture Notes in Computer Science, vol. 5502.
\newblock Springer Berlin Heidelberg.

\bibitem[\protect\citename{Reynolds \& Plotkin, }1993]{ReynoldsP93}
Reynolds, John~C., \& Plotkin, Gordon~D. (1993).
\newblock On functors expressible in the polymorphic typed lambda calculus.
\newblock {\em Information and computation}, {\bf 105}(1), 1--29.

\bibitem[\protect\citename{Swierstra, }2008]{swierstra2008:la-carte}
Swierstra, Wouter. (2008).
\newblock Data types \`a la carte.
\newblock {\em Journal of functional programming}, {\bf 18}(4), 423--436.

\end{thebibliography}

\appendix

\allowdisplaybreaks

\section{Ends}
\label{sec:ends}

Ends are a special type of limit. The limit for a functor $F : \cC \to
\cD$ is a universal natural transformation $K_D\to F$ (the universal
cone to $F$) from the functor which is constantly $D$, for a
$D\in\cD$, into the functor $F$. The end for a functor $F : \op{\cC}
\times \cC \to \cD$ arises as a \emph{dinatural} transformation
$K_D\to F$ (the universal wedge).

\begin{defn}\label{def:dinatural}
  A \emph{dinatural transformation} $\alpha : F \to G$ between
  functors $F, G : \op{\cC} \times \cC \to \cD$ is a family of
  morphisms of the form $\alpha_C : F(C,C) \to G(C,C)$, such that for
  every morphism $f : C \to C'$ the following diagram commutes.
  
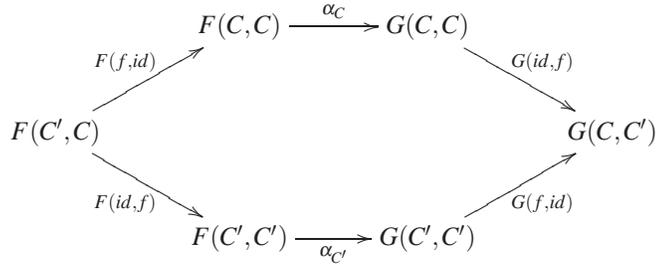
\[
  \xymatrix@C+=10mm{
    & F(C,C) \ar[r]^{\alpha_C}      & G(C,C) \ar[rd]^{G(id,f)}   & \\
    F(C',C) \ar[ru]^{F(f,id)} \ar[rd]_{F(id,f)} &                            &                           & G(C,C') \\
    & F(C',C') \ar[r]_{\alpha_{C'}}  & G(C',C') \ar[ru]_{G(f,id)}  &
  }
  \]%
\end{defn}


Differently from natural transformations, dinatural transformations
are not closed under composition.


\begin{defn}
  A \emph{wedge} from an object $V \in \cD$ to a functor $F : \op{\cC} \times
  \cC \to \cD$ is a dinatural transformation from the constant functor
  $K_V : \op{\cC} \times \cC \to \cD$ to $F$. Explicitly, an object
  $V$ together with a family of morphisms $\alpha_X : V \to F(X,X)$
  such that for each $f : C \to C'$ the following diagram commutes.
  
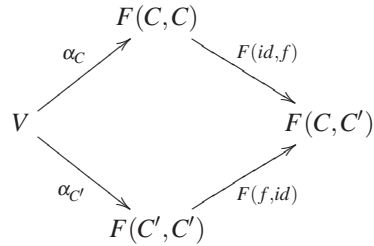
\[
  \xymatrix{
                                  & F(C,C) \ar[dr]^{F(id,f)}    & \\
    V \ar[ur]^{\alpha_C} \ar[dr]_{\alpha_{C'}}  &                             & F(C,C') \\
                                  & F(C',C') \ar[ur]_{F(f,id)}  &
  }
  \]%

\end{defn}

Whereas a limit is a final cone, an \emph{end} is a final wedge.

\begin{defn}
  The \emph{end} of a functor $F : \op{\cC} \times \cC \to \cD$ is a
  final wedge for $F$.  Explicitly, it is an object $\int_A
  F(A,A)\in\cD$  together with a family of morphisms
  $\omega_C : \int_A F(A,A) \to F(C,C)$ such that the diagram
  
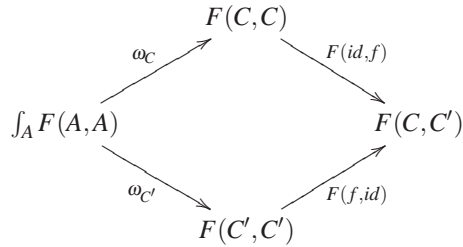
\[
  \xymatrix{
                                  & F(C,C) \ar[dr]^{F(id,f)}    & \\
    \int_A F(A,A) \ar[ur]^{\omega_C} \ar[dr]_{\omega_{C'}}  &                             & F(C,C') \\
                                  & F(C',C') \ar[ur]_{F(f,id)}  &
  }
  \]%
  commutes for each $f : C \to C'$, and such that for every wedge from
  $V\in\cD$, given by a family of morphisms $\gamma_c : V \to F(C,C)$
  such that $F(id,f) \circ \gamma_c = F(f,id) \circ \gamma_c'$ for
  every $f : C \to C'$, there exists a unique morphism $! : V \to \int_A F(A,A)$ such that the following diagram commutes.
  \[
  \xymatrix{
                                               & &                                  & F(C,C) \ar[dr]^{F(id,f)}    & \\
    V \ar[urrr]^{\gamma_C} \ar[drrr]_{\gamma_{C'}} \ar@{-->}[rr]^{!} & & \int_A F(A,A) \ar[ur]_{\omega_C} \ar[dr]^{\omega_{C'}}  &                             & F(C,C') \\
                                               & &                                  & F(C',C') \ar[ur]_{F(f,id)}  &
  }
  \]%
\end{defn}

\begin{remark}\label{remark:ends_as_limits}
  When $\cC$ is small and $\cD$ is small-complete, an end over a
  functor $\cC\times\op{\cC}\to\cD$ can be reduced to an ordinary
  limit~\cite{macLaneS:catwm}. As a consequence, the Hom functor
  preserves ends: for every $D\in\cD$,
\[
\Hom{\cD}{D}{\int_A F(A,A)} 
   \quad=\quad 
\int_A \Hom{\cD}{D}{F(A,A)}.
\]%
\end{remark}







%

\section{Generalised lens representation theorem}
\label{thm:generalised-lens-representation}


For all the propositions below assume we have two small monoidal
categories of endofunctors, $(\cat{E},\Id, \cdot, \alpha, \lambda,
\rho)$ and $(\cat{F}, \Id, \cdot, \alpha', \lambda', \rho')$,
$\cat{E}$ is a subcategory of endofunctors over a base category $\cC$,
and $\cat{F}$ is a subcategory of endofunctors over a base category
$\cD$, and where the monoidal operation is composition of endofunctors
(written $F \cdot G$) and with the identity functor, $\Id$, as the
identity.  Also assume we have an adjunction $(-)^* \dashv U :
\cat{E}\rightharpoonup \cat{F}$, such that $U$ is strict
monoidal\footnote{ These propositions still hold under the assumption
  that $U$ is a strong monoidal functor.  In order to avoid excessive
  notation we use the simplifying assumption that $U$ is strict.}
(i.e., $U\,\Id = \Id$, $U(F \cdot G)=UF \cdot UG$, $U \lambda'_X =
\lambda_{U X}$, etc.).

To reduce notational clutter, in this section we work directly with natural
transformations.  Rather that writing the counit of a parameterised comonad
as a family of arrows $\varepsilon_{a,X} : C_{a,a}X \to X$ as we did in
Section~\ref{sec:pcomonads_and_lenses}, we will write it as a family of natural
transformations, $\varepsilon_{a} : C_{a,a} \to \Id$.  Similarly, instead of
writing comultiplication as $\delta_{a,b,c,X} : C_{a,c}X \to C_{a,b}(C_{b,c}X)$
we will write $\delta_{a,b,c} : C_{a,c} \to C_{a,b} \cdot C_{b,c}$, and so
forth.

 \begin{prop}
   Let $(C, \varepsilon^C, \delta^C)$ be a
   $\cP$-parameterised comonad on $\cC$, such that for every $a, b
   : \cP$, we have an endofunctor $C_{a,b} : \cat{E}$.
   Then $(C^*,\varepsilon^{C^*},\delta^{C^*})$ is a $\cP$-parameterised comonad
   on $\cD$ where
 $$
 \begin{array}{lcl}
 \varepsilon^{C^*}_{a}  & : & C^*_{a,a} \to \Id \\
 \varepsilon^{C^*}_{a} & = & \radj{\varepsilon^{C}_a} \\
 \\
 \delta^{C^*}_{a,b,c} & : &  C^*_{a,c} \to C^*_{a,b} \cdot C^*_{b,c} \\
 \delta^{C^*}_{a,b,c} & = & \radj{(\eta_{C_{a,b}}\cdot\eta_{C_{b,c}}) \circ \delta^{C}_{a,b,c}}
 \end{array}
 $$
 The tensor $\cdot$ in the term corresponds to horizontal composition
 of natural transformations.
 \end{prop}

\begin{proof}
 The first parameterised comonad law is:
 \[
 \lambda_{C_{a,b}} \circ (\varepsilon^C_a \cdot id) \circ \delta^C_{a,a,b}=id : 
 \Hom{\cat{E}}{C_{a,b}}{C_{a,b}}
 \]%
We check that:
 \[
 \lambda'_{C^*_{a,b}} \circ (\varepsilon^{C^*}_a \cdot id) \circ \delta^{C^*}_{a,a,b}=id : 
 \Hom{\cat{F}}{C^*_{a,b}}{C^*_{a,b}}
 \]%
 \allowdisplaybreaks 
\begin{flalign*}
    &  \lambda'_{C^*_{a,b}} \circ (\varepsilon^{C^*}_a \cdot id) 
                      \circ \delta^{C^*}_{a,a,b}
    & \\
 =  & \comment{Definition of $\delta^{C^*}$} 
    & \\ 
    & \lambda'_{C^*_{a,b}} \circ (\varepsilon^{C^*}_a \cdot id) \circ
      \radj{(\eta_{C_{a,a}} \cdot \eta_{C_{a,b}}) \circ \delta^C_{a,a,b}} 
    & \\
 =  & \comment{Eq.~\ref{eq:adj-nat2}} 
     & \\
    & \radj{U\lambda'_{C^*_{a,b}} \circ U(\varepsilon^{C^*}_a \cdot id) \circ
      (\eta_{C_{a,a}} \cdot \eta_{C_{a,b}}) \circ \delta^C_{a,a,b}}
     & \\
 =  & \comment{\ensuremath{U} is strict monoidal.}
    & \\
    & \radj{\lambda_{UC^*_{a,b}} \circ (U\,\varepsilon^{C^*}_a \cdot id) \circ
    (\eta_{C_{a,a}} \cdot \eta_{C_{a,b}}) \circ \delta^C_{a,a,b}}
   & \\
 =  & \comment{ Bifunctor \ensuremath{\cdot}, definition of $\varepsilon^{C^*}$}
   & \\
    & \radj{\lambda_{UC^*_{a,b}} \circ 
      ((U\,\radj{\varepsilon^C_a} \circ \eta_{C_{a,a}}) \cdot \eta_{C_{a,b}}) 
      \circ \delta^C_{a,a,b}}
   & \\
 =  & \comment{ Eq~\ref{eq:adjuncts-from-unit-counit}}
   & \\
    & \radj{\lambda_{UC^*_{a,b}} \circ 
      (\ladj{\radj{\varepsilon^C_a}} \cdot \eta_{C_{a,b}}) 
      \circ \delta^C_{a,a,b}}
   & \\
 =  & \comment{ isomorphism}
   & \\
    & \radj{\lambda_{UC^*_{a,b}} \circ 
      (\varepsilon^C_a \cdot \eta_{C_{a,b}}) 
      \circ \delta^C_{a,a,b}}
   & \\
 =  & \comment{ naturality of \ensuremath{\lambda}}
   & \\
    & \radj{\eta_{C_{a,b}} \circ \lambda_{C_{a,b}} \circ 
      (\varepsilon^C_a \cdot id) 
      \circ \delta^C_{a,a,b}} 
   & \\
 =  & \comment{ first parameterised comonad law}
   & \\
    & \radj{\eta_{C_{a,b}}}
   & \\
 =  & \comment{Eq.~\ref{eq:eta-epsilon-from-adjuncts}}
   & \\
    & \radj{\ladj{id}}
   & \\
 =  & \comment{ isomorphism}
   & \\
    & id   
\end{flalign*}
  
For the second parameterised comonad law we proceed in a similar way
to the first.

 The third parameterised comonad law states
 \[
  \alpha_{C_{a,b},C_{b,c},C_{c,d}} \circ (\delta^C_{a,b,c} \cdot id ) \circ \delta^C_{a,c,d} 
 =       
  (id \cdot \delta^C_{b,c,d}) \circ \delta^C_{a,b,d}
 : \Hom{\cat{E}}{C_{a,d}}{C_{a,b} \cdot (C_{b,c} \cdot C_{c,d})}
 \]%
 Let us prove that
 \[
  \alpha'_{C^*_{a,b},C^*_{b,c},C^*_{c,d}} \circ (\delta^{C^*}_{a,b,c} \cdot id ) \circ \delta^{C^*}_{a,c,d} 
 =       
  (id \cdot \delta^{C^*}_{b,c,d}) \circ \delta^{C^*}_{a,b,d}
 : \Hom{\cat{F}}{C^*_{a,d}}{C^*_{a,b} \cdot (C^*_{b,c} \cdot C^*_{c,d})}
 \]%
 \begin{flalign*}
    &  \alpha' \circ (\delta^{C^*}_{a,b,c} \cdot id ) \circ \delta^{C^*}_{a,c,d}
    & \\
 =  & \comment{Definition of $\delta^{C^*}$} 
    & \\ 
    &  \alpha' \circ (\delta^{C^*}_{a,b,c} \cdot id ) \circ
       \radj{(\eta_{C_{a,c}} \cdot \eta_{C_{c,d}}) \circ \delta^C_{a,c,d}}
    & \\
 =  & \comment{Eq.~\ref{eq:adj-nat2}, \ensuremath{U} strict monoidal} 
     & \\
    &  \radj{\alpha \circ (U\delta^{C^*}_{a,b,c} \cdot id ) \circ 
       (\eta_{C_{a,c}} \cdot \eta_{C_{c,d}}) \circ \delta^C_{a,c,d}}
    & \\
 =  & \comment{\ensuremath{\cdot} bifunctor} 
     & \\
    &  \radj{\alpha \circ ((U\delta^{C^*}_{a,b,c} \circ \eta_{C_{a,c}}) \cdot 
       \eta_{C_{c,d}}) \circ \delta^C_{a,c,d}}
    & \\
 =  & \comment{ Eq~\ref{eq:adjuncts-from-unit-counit}}
   & \\
    &  \radj{\alpha \circ (\ladj{
              \delta^{C^*}_{a,b,c}} \cdot \eta_{C_{c,d}}) \circ \delta^C_{a,c,d}}
   & \\
 =  & \comment{Definition of $\delta^{C^*}$} 
    & \\ 
    &  \radj{\alpha \circ (\ladj{
 \radj{(\eta_{C_{a,b}} \cdot \eta_{C_{b,c}}) \circ \delta^C_{a,b,c}}}
       \cdot  \eta_{C_{c,d}}) \circ \delta^C_{a,c,d}}
   & \\
 =  & \comment{ isomorphism}
   & \\
    &  \radj{\alpha \circ((\eta_{C_{a,b}} \cdot \eta_{C_{b,c}}) \circ 
                    \delta^C_{a,b,c}) \cdot \eta_{C_{c,d}}) 
                \circ \delta^C_{a,c,d}}
   & \\
 =  & \comment{\ensuremath{\cdot} bifunctor} 
     & \\
    & \radj{
         \alpha \circ ((\eta_{C_{a,b}} \cdot \eta_{C_{b,c}}) \cdot
         \eta_{C_{c,d}}) \circ (\delta^C_{a,b,c} \cdot id) 
                \circ \delta^C_{a,c,d}
        }
    & \\
 =  & \comment{naturality of \ensuremath{\alpha}} 
     & \\
    & \radj{
         ((\eta_{C_{a,b}} \cdot (\eta_{C_{b,c}} \cdot  \eta_{C_{c,d}})) \circ
          \alpha \circ (\delta^C_{a,b,c} \cdot id) \circ \delta^C_{a,c,d}
        }
    & \\
 =  & \comment{ third parameterised comonad law}
   & \\
    & \radj{
         ((\eta_{C_{a,b}} \cdot (\eta_{C_{b,c}} \cdot  \eta_{C_{c,d}})) 
          \circ 
        (id \cdot \delta^C_{b,c,d}) \circ \delta^C_{a,b,d}
        }
    & \\
 =  & \comment{\ensuremath{\cdot} bifunctor} 
    & \\
    & \radj{ (\eta_{C_{a,b}} \cdot 
    ((\eta_{C_{b,c}} \cdot \eta_{C_{c,d}}) \circ \delta^C_{b,c,d})) 
     \circ \delta^C_{a,b,d}}
   & \\
 =  & \comment{ isomorphism}
   & \\
    & \radj{ (\eta_{C_{a,b}} \cdot \ladj{
         \radj{(\eta_{C_{b,c}} \cdot \eta_{C_{c,d}}) \circ \delta^C_{b,c,d}}
          }) 
     \circ \delta^C_{a,b,d}}
   & \\ 
 =  & \comment{Definition of $\delta^{C^*}$} 
 & \\
    & \radj{ (\eta_{C_{a,b}} \cdot \ladj{\delta^{C^*}_{b,c,d}}) 
     \circ \delta^C_{a,b,d}}
   & \\
 =  & \comment{ Eq~\ref{eq:adjuncts-from-unit-counit}}
   & \\
    & \radj{ (\eta_{C_{a,b}} \cdot (U\delta^{C^*}_{b,c,d}\circ \eta_{C_{b,d}})) 
     \circ \delta^C_{a,b,d}}
    & \\
 =  & \comment{\ensuremath{\cdot} bifunctor} 
     & \\
    & \radj{ (id \cdot U\delta^{C^*}_{b,c,d}) \circ 
       (\eta_{C_{a,b}} \cdot \eta_{C_{b,d}}) \circ \delta^C_{a,b,d}}
    & \\
 =  & \comment{Eq.~\ref{eq:adj-nat2}, \ensuremath{U} strict monoidal} 
     & \\
    & (id \cdot \delta^{C^*}_{b,c,d}) \circ 
       \radj{(\eta_{C_{a,b}} \cdot \eta_{C_{b,d}}) \circ \delta^C_{a,b,d}}
   & \\ 
 =  & \comment{Definition of $\delta^{C^*}$} 
 & \\
   &  (id \cdot \delta^{C^*}_{b,c,d}) \circ \delta^{C^*}_{a,b,d} 
 \end{flalign*}%
\end{proof}

 \begin{prop}
   Let $(D, \varepsilon^D, \delta^D)$ be a
   $\cP$-parameterised comonad on $\cD$, such that for every $a, b
   : \cP$, we have an endofunctor $D_{a,b} : \cat{F}$.
   Then $(U\,D,\varepsilon^{U\,D},\delta^{U\,D})$ is a
   $\cP$-parameterised comonad on $\cC$ where
 $$
 \begin{array}{lcl}
 \varepsilon^{U\,D}_{a}  & : & U\,D_{a,a} \to \Id \\
 \varepsilon^{U\,D}_{a} & = & U\varepsilon^D_a \\
 \\
 \delta^{U\,D}_{a,b,c} & : & U\,D_{a,c} \to U\,D_{a,b} \cdot U\,D_{b,c} \\
 \delta^{U\,D}_{a,b,c} & = & U\delta^D_{a,b,c}
 \end{array}
 $$
 \end{prop}
\begin{proof}
  The laws of a parameterised comonad follow directly from the fact
  that $U$ is a strict monoidal functor.
\end{proof}

\begin{prop}[Generalised lens representation (Theorem~\ref{theorem:generalized-lens-representation}]\label{prop:glr}
  Given a functor $K : \cP \to Set$, define $R^{(K)}_{a,b} X = K a \times (K b \tto X) : \cP\times\op\cP\times Set \to
Set$ as the parameterised comonad with counit $\varepsilon^{R^{(K)}}$ and
comultiplication $\delta^{R^{(K)}}$ as defined in Example~\ref{example:RK}.
Assume that $R^{(K)}_{a,b} : \cat{E}$ for every $a$ and $b$.  Then 
\begin{enumerate}
\item $UR^{(K)*}$ is a parameterised comonad and
\item given a functor $J : \cP \to Set$, then the families
  $k_{a,b} : Ja \tto UR^{(K)*}_{a,b}(Jb)$ which form the
  $UR^{(K)*}$-coalgebras $(J,k)$ are isomorphic to the families of ends
\[
  \int_{\F:\cat{F}} (Ka \tto U\F (Kb)) \tto Ja \tto U\F (Jb)
\]%
which satisfy the linearity and unity laws.
\end{enumerate}
\end{prop}
\begin{proof}
  The previous two propositions entail that $UR^{(K)*}$ is a
  parameterised comonad with the following counit and
  comultiplication.
 $$
 \begin{array}{lcl}
 \varepsilon^{UR^{(K)*}}_{a}  & : & UR^{(K)*}_{a,a} \to \Id \\
 \varepsilon^{UR^{(K)*}}_{a} & = & U\radj{\varepsilon^{R^{(K)}}_a} \\
 \\
 \delta^{UR^{(K)*}}_{a,b,c} & : & UR^{(K)*}_{a,c} \to UR^{(K)*}_{a,b} \cdot UR^{(K)*}_{b,c} \\
 \delta^{UR^{(K)*}}_{a,b,c} & = & U\radj{(\eta_{R^{(K)}_{a,b}}\cdot\eta_{R^{(K)}_{b,c}}) \circ \delta^{R^{(K)}}_{a,b,c}}
 \end{array}
 $$%
The unary representation theorem (Theorem~\ref{thm:representation1}) entails the isomorphism
$$
  Ja \tto UR^{(K)*}_{a,b}(Jb) 
        \quad \cong \quad 
     \int_{\F:\cat{F}} (Ka \tto U\F (Kb)) \tto Ja \tto U\F (Jb)
$$%
witnessed by the following functions
\[
\begin{array}{l@{\,}c@{\,}l}
\gamma & : & (\int_{\F}  (Ka \tto U\F (Kb)) \tto Ja \tto U\F (Jb)) \to (Ja \tto UR^{(K)*}_{a,b}\,(Jb))
\\
\gamma(h) & = & h_{R^{(K)*}_{a,b}}\,(\alpha^{-1}_{UR^{(K)*}_{a,b}}(\eta_{R^{(K)}_{a,b}}))
\\
\\
\gamma^{-1} & : &  (Ja \tto UR^{(K)*}_{a,b}\,(Jb)) \to \int_{\F} (Ka \tto U\F (Kb)) \tto (Ja \tto U\F (Jb))\\
\gamma^{-1}(k) & = & \tau \quad\text{where~}
\begin{array}[t]{l@{\,}c@{\,}l}
\tau_{\F} & : & (Ka \tto U\F (Kb)) \tto Ja \tto U\F (Jb) \\
\tau_{\F}(g) & = & U\radj{\alpha_{U\F}(g)}_{Jb} \circ k
\end{array} \\
\end{array}
\]%
All that remains is to show that $k_{a,b}$ satisfies the coalgebra
laws if and only if $\gamma^{-1}(k_{a,b})$ satisfies the linearity and
unity laws.

First we prove two lemmas:
\begin{lemma}\label{lemma:gen-tech1}
  For all $F,G : \cat{F}$ and $f : Ka \to UF(Kb)$ and $g : Kb \to
  UG(Kc)$ we have that
\[
  \gamma^{-1}(k_{a,c})_{F \cdot G}(UFg \circ f) \quad = \quad
  U(\radj{\alpha_{UF}(f)} \cdot \radj{\alpha_{UG}(g)})_{Jc} \circ \delta^{UR^{(K)*}}_{a,b,c,Jc} \circ k_{a,c}
\]%
\end{lemma}
\begin{proof}
 \begin{flalign*}
    & \gamma^{-1}(k_{a,c})_{F \cdot G}(UFg \circ f)
    & \\
 =  & \comment{Definition of $\gamma^{-1}$} 
   \\
    & U\radj{\alpha_{UF \cdot UG}(UFg \circ f)}_{Jc} \circ k_{a,c}
    \\
 =  & \comment{Proposition~\ref{prop:delta-alpha}(b)} 
   \\
    & U\radj{(\alpha_{UF}(f) \cdot \alpha_{UG}(g)) \circ \delta^{R^{(K)}}_{a,b,c}}_{Jc} \circ k_{a,c}
   \\
 =  & \comment{isomorphism} 
   \\
    & U\radj{(\ladj{\radj{\alpha_{UF}(f)}} \cdot \ladj{\radj{\alpha_{UG}(g)}}) \circ \delta^{R^{(K)}}_{a,b,c}}_{Jc} \circ k_{a,c}
   \\
 =  & \comment{Eq~\ref{eq:adjuncts-from-unit-counit}} 
   \\
    & U\radj{((U\radj{\alpha_{UF}(f)} \circ \eta_{R^{(K)}_{a,b}}) \cdot
              (U\radj{\alpha_{UG}(g)} \circ \eta_{R^{(K)}_{b,c}}))
      \circ \delta^{R^{(K)}}_{a,b,c}}_{Jc} \circ k_{a,c}
   \\
 =  & \comment{$\cdot$ bifunctor and U is strict} 
   \\
    & U\radj{U(\radj{\alpha_{UF}(f)} \cdot \radj{\alpha_{UG}(g)}) \circ
             (\eta_{R^{(K)}_{a,b}} \cdot \eta_{R^{(K)}_{b,c}})
      \circ \delta^{R^{(K)}}_{a,b,c}}_{Jc} \circ k_{a,c}
   \\
 =  & \comment{Eq~\ref{eq:adj-nat2} and U is strict} 
   \\
    & (U(\radj{\alpha_{UF}(f)} \cdot \radj{\alpha_{UG}(g)}) \circ
       U\radj{(\eta_{R^{(K)}_{a,b}} \cdot \eta_{R^{(K)}_{b,c}})
      \circ \delta^{R^{(K)}}_{a,b,c}})_{Jc} \circ k_{a,c}
   \\
 =  & \comment{Definition of $\delta^{UR^{(K)*}}$} 
   \\
    & U(\radj{\alpha_{UF}(f)} \cdot \radj{\alpha_{UG}(g)})_{Jc} \circ \delta^{UR^{(K)*}}_{a,b,c,Jc} \circ k_{a,c}
 \end{flalign*}
\end{proof}

We note that Lemma~\ref{lemma:tech1} follows from Lemma~\ref{lemma:gen-tech1}
by considering the identity adjunction between $\cat{E}$ and itself.

\begin{lemma}\label{lemma:gen-tech2}
  For all $F,G : \cat{F}$ and $f : Ka \tto UF(Kb)$ and $g : Kb \tto
  UG(Kc)$ we have that
\[
  UF(\gamma^{-1}(k_{b,c})_{G}(g)) \circ \gamma^{-1}(k_{a,b})_{F}(f) \quad = \quad
  U(\radj{\alpha_{UF}(f)} \cdot \radj{\alpha_{UG}(g)})_{Jc} \circ UR^{(K)*}_{a,b}(k_{b,c}) \circ k_{a,b}
\]%
\end{lemma}
\begin{proof}%
 \begin{flalign*}
    & UF(\gamma^{-1}(k_{b,c})_{G}(g)) \circ \gamma^{-1}(k_{a,b})_{F}(f) 
   & \\
 =  & \comment{Definition of $\gamma^{-1}$} 
   \\
    & UF(U\radj{\alpha_{UG}(g)}_{Jc} \circ k_{b,c}) \circ U\radj{\alpha_{UF}(f)}_{Jb} \circ k_{a,b} 
   \\
 =  & \comment{$UF$ is a functor} 
   \\
    & UF(U\radj{\alpha_{UG}(g)}_{Jc}) \circ UF(k_{b,c}) \circ U\radj{\alpha_{UF}(f)}_{Jb} \circ k_{a,b} 
   \\
 =  & \comment{$U\radj{\alpha_{UF}(f)}$ is natural} 
   \\
    & UF(U\radj{\alpha_{UG}(g)}_{Jc}) \circ U\radj{\alpha_{UF}(f)}_{UR^{(K)*}_{b,c}(Jc)} \circ UR^{(K)*}_{a,b}(k_{b,c}) \circ k_{a,b} 
   \\
 =  & \comment{Definition of $\cdot$}
   \\
    & (U\radj{\alpha_{UF}(f)} \cdot U\radj{\alpha_{UG}(g)})_{Jc} \circ UR^{(K)*}_{a,b}(k_{b,c}) \circ k_{a,b}
   \\
 =  & \comment{$U$ is strict} 
   \\
    & U(\radj{\alpha_{UF}(f)} \cdot \radj{\alpha_{UG}(g)})_{Jc} \circ UR^{(K)*}_{a,b}(k_{b,c}) \circ k_{a,b}
 \end{flalign*}%
\end{proof}

We note that Lemma~\ref{lemma:tech2} follows from
Lemma~\ref{lemma:gen-tech2} by considering the identity adjunction
between $\cat{E}$ and itself.

The linearity law for the image of $\gamma^{-1}$ states
\[
  \forall F, G, f, g.
  \gamma^{-1}(k_{a,c})_{F \cdot G}(UFg \circ f) = 
  UF(\gamma^{-1}(k_{b,c})_{G}(g)) \circ \gamma^{-1}(k_{a,b})_{F}(f)
\]%

By the previous two lemmas, this linearity law is equivalent to stating that $  \forall F, G, f, g$
\[
  U(\radj{\alpha_{UF}(f)} \cdot \radj{\alpha_{UG}(g)})_{Jc} \circ \delta^{UR^{(K)*}}_{a,b,c,Jc} \circ k_{a,c}
  = 
  U(\radj{\alpha_{UF}(f)} \cdot \radj{\alpha_{UG}(g)})_{Jc} \circ UR^{(K)*}_{a,b}(k_{b,c}) \circ k_{a,b}
\]%
With this reformulation we see that the comultiplication-coalgebra law,
$$
  \delta^{UR^{(K)*}}_{a,b,c,Jc} \circ k_{a,c} \quad = \quad UR^{(K)*}_{a,b}(k_{b,c}) \circ k_{a,b}
$$
trivially implies the linearity law.  To derive the
comultiplication-coalgebra law from the linearity law consider the
instance where 
$F =R^{(K)}_{a,b}$, 
$f =\alpha^{-1}_{UR^{(K)*}_{a,b}}(\eta_{R^{(K)}_{a,b}})$, 
$G = R^{(K)}_{b,c}$, and 
$g = \alpha^{-1}_{UR^{(K)*}_{b,c}}(\eta_{R^{(K)}_{b,c}})$.  
In this case we have
 \begin{flalign*}
    & U(\radj{\alpha_{UF}(f)} \cdot \radj{\alpha_{UG}(g)})
    & \\
 =  & \comment{definition of $f$ and $g$} 
    \\
    & U(\radj{\alpha_{UR^{(K)*}_{a,b}}(\alpha^{-1}_{UR^{(K)*}_{a,b}}(\eta_{R^{(K)}_{a,b}}))} \cdot
        \radj{\alpha_{UR^{(K)*}_{b,c}}(\alpha^{-1}_{UR^{(K)*}_{b,c}}(\eta_{R^{(K)}_{b,c}}))})
    \\
 =  & \comment{isomorphism} 
    \\
    & U(\radj{\eta_{R^{(K)}_{a,b}}} \cdot \radj{\eta_{R^{(K)}_{b,c}}})
    \\
 =  & \comment{Eq.~\ref{eq:eta-epsilon-from-adjuncts}}
    \\
    & U(\radj{\ladj{id}} \cdot \radj{\ladj{id}})
    \\
 =  & \comment{isomorphism} 
    \\
    & U(id \cdot id)
    \\
 =  & \comment{identity} 
   \\
    & id
 \end{flalign*}
and then the comultiplication-coalgebra law follows.

The unity law for the image of $\gamma^{-1}$ states
\[
  \gamma^{-1}(k_{a,a})_{\Id}(id) \quad = \quad id : J a \tto J a
\]%
The counit-coalgebra law states
\[
  \varepsilon^{UR^{(K)*}} \circ k_{a,a} \quad = \quad id : J a \tto J a
\]%
Therefore, in order to show that these laws are equivalent, it
suffices to prove the following.
\[
  \gamma^{-1}(k_{a,a})_{\Id}(id) \quad = \quad \varepsilon^{UR^{(K)*}} \circ k_{a,a}
\]%
 \begin{flalign*}
    & \gamma^{-1}(k_{a,a})_{\Id}(id)
    & \\
 =  & \comment{definition of $\gamma^{-1}$}
    \\
    & U(\radj{\alpha_{\Id}(id)})(J a) \circ k_{a,a}
    \\
 =  & \comment{Proposition~\ref{prop:epsilon-alpha}} 
    \\
    & U(\radj{\varepsilon^{R^{(K)}}_a})(J a) \circ k_{a,a}
    \\
 =  & \comment{Definition of $\varepsilon^{UR^{(K)*}}$}
    \\
    & \varepsilon^{UR^{(K)*}} \circ k_{a,a}
 \end{flalign*}
\end{proof}

%

\end{document}